\documentclass[11pt]{article}
 \usepackage{amssymb,amsmath}
\usepackage{amsthm}
\usepackage{graphicx,epsfig,mathrsfs}
 \usepackage{graphicx}
\usepackage{boxedminipage}
\usepackage{color}

\usepackage{vmargin}
 
\newcommand{\defparprob}[4]{
  \vspace{1mm}
\noindent\fbox{
  \begin{minipage}{0.96\textwidth}
  \begin{tabular*}{\textwidth}{@{\extracolsep{\fill}}lr} #1  & {\bf{Parameter:}} #3 \\ \end{tabular*}
  {\bf{Input:}} #2  \\
  {\bf{Question:}} #4
  \end{minipage}
  }
  \vspace{1mm}
}

\usepackage{algorithm}
\usepackage{algorithmic}

\usepackage{graphicx}
\usepackage{amsfonts}
\usepackage{amssymb}
\usepackage{amsmath}
\usepackage{amssymb}
\usepackage{latexsym}

\newtheorem{theorem}{Theorem}
\newtheorem{observation}{Observation}
\newtheorem{proposition}{Proposition}

\newtheorem{lemma}{Lemma}

\newtheorem{fact}{Fact}

\newtheorem{definition}{Definition}

\newcommand{\tw}{{\mathbf{tw}}}

\newcommand{\tHmf}{$H$-topological-minor-free}
\newcommand{\Hmf}{$H$-minor-free}

\newcommand{\tDS}{{\texttt{\sc DS}}}
\newcommand{\tCDS}{{\texttt{\sc CDS}}}

\newcommand{\cO}{\mathcal{O}}

\newcommand{\h}[1]{\end{document}}

\usepackage{todonotes}
\newcommand{\mar}[1]{#1}
\newcommand{\term}[1]{#1}

\begin{document}

\title{Kernels for (connected) Dominating Set on  graphs with \\ Excluded Topological subgraphs\thanks{Preliminary versions of this paper appeared in SODA 2012 and STACS 2013. }}
\author{Fedor V. Fomin\thanks{University of Bergen, Norway, \texttt{fomin@ii.uib.no}.  The research was supported by the European Research Council through ERC Grant Agreement n. 267959.}
\and 
Daniel Lokshtanov\thanks{University of Bergen, Norway, \texttt{daniello@ii.uib.no}. The research was supported by the Bergen Research Foundation and the University of Bergen through project ``BeHard'.}
\and
Saket Saurabh\thanks{The Institute of Mathematical Sciences, CIT Campus, Chennai, India, \texttt{saket@imsc.res.in}. 
The research was supported by the European Research Council through Starting Grant 306992 ``Parameterized Approximation'.}
\and 
Dimitrios M. Thilikos\thanks{Department of Mathematics, National and Kapodistrian University of Athens, Athens, Greece,  \texttt{sedthilk@thilikos.info}. Co-financed by the E.U. (European Social Fund - ESF) and Greek national funds through the Operational Program ``Education and Lifelong Learning'' of the National Strategic Reference Framework (NSRF) - Research Funding Program: ``Thales. Investing in knowledge society through the European Social Fund''.}~\thanks{AlGCo project-team, CNRS, LIRMM, Montpellier, France}}
\date{}
\maketitle
\thispagestyle{empty}

\begin{abstract}
\vspace{1mm}

\noindent We give the first linear kernels for the  {\sc Dominating Set}  and  {\sc Connected Dominating Set}  problems on graphs excluding a fixed graph $H$ as a topological minor.  In other words, we prove the existence of polynomial time algorithms that, for a  given  \tHmf\,  graph $G$ and  a positive integer $k$, 
output an \tHmf\,  graph $G'$   on  $\cO(k)$ vertices such that $G$ has a (connected) dominating set of size $k$ if and only if $G'$ has one.

Our results extend the known classes of graphs on which the {\sc Dominating Set}  and  {\sc Connected Dominating Set}  problems admit linear kernels.  Prior to our work, it was known that these problems admit linear kernels on graphs excluding a fixed  apex graph $H$ as a minor.   Moreover, for   {\sc Dominating Set},  a kernel of size  $k^{c(H)}$, where $c(H)$ is a  constant depending  on the size of $H$, follows from a more general result on the kernelization of {\sc Dominating Set} on graphs of bounded degeneracy.  Alon and Gutner explicitly asked  whether one can obtain a linear kernel for {\sc Dominating Set} on \Hmf\,  graphs. We answer this question in the affirmative and in fact prove a more general result. 
For  {\sc Connected Dominating Set}  no  polynomial kernel  even on \Hmf\,  graphs was known prior to our work.  On the negative side, it is known that {\sc Connected Dominating Set}  on $2$-degenerated graphs does not admit a polynomial kernel unless \textsf{coNP} $\subseteq$  \textsf{NP/poly}.

Our kernelization algorithm is based on a non-trivial combination of the following  ingredients
 \begin{itemize} 
\item The structural theorem of Grohe and Marx [STOC 2012] for graphs excluding a fixed graph $H$ as a topological minor;
\item A novel notion of protrusions, different than the one defined in [FOCS 2009];
\item Our results are based on  a  generic reduction rule that produces an equivalent instance (in case the input graph is \Hmf)  of the problem, with treewidth 
$\cO(\sqrt{k})$. 
The application of this rule  in a divide-and-conquer fashion, together with the new notion of protrusions,  gives us the linear kernels. 

\end{itemize}
A protrusion in a graph [FOCS 2009] is a subgraph of constant treewidth which is separated from the rest of the graph by at most a constant number of vertices. In our variant of protrusions, instead of stipulating that the subgraph be of constant 
{\em treewidth}, we ask that it contains a {\em constant number of vertices from a solution}. We believe that this new take on protrusions would be useful for other graph problems and in different algorithmic settings. 
\end{abstract}


\maketitle


\section{Introduction}
 {\em Kernelization} is well established subarea of parameterized complexity. 
 A parameterized problem is said to admit a {\em polynomial kernel} if there is a polynomial time algorithm (the degree of polynomial being independent of the parameter $k$), called a {\em kernelization} algorithm, that reduces the input instance down to an instance with size bounded by a polynomial $p(k)$ in $k$, while preserving the answer. This reduced instance is called a {\em $p(k)$ kernel} for the problem. If the size of the kernel is $\cO(k)$, then we call it a {\em linear kernel} (for a more formal definition, see Section~\ref{sec:defs_and_nots}). Kernelization has turned out to  be an interesting computational approach both from practical and theoretical perspectives. There are many real-world applications where even very simple preprocessing can be surprisingly effective, leading to significant reductions  in the size of the input. Kernelization is a natural tool not only for  measuring the quality of preprocessing rules proposed for  specific problems but also for designing new powerful preprocessing algorithms. From the theoretical perspective, kernelization provides a deep insight into the hierarchy of parameterized problems in  {\sf FPT}, the most interesting class of parameterized problems.  There are also interesting links  
  between lower bounds on the sizes of kernels and classical computational complexity  \cite{BDFH08,Dell:2010sh,DruckerA12}.

The {\sc Dominating Set} (\tDS) problem
 together with its numerous variants, is one of the most  classical and well-studied problems in algorithms and combinatorics~\cite{HaynesHS98}. 
  In the {\sc Dominating Set} (\tDS) problem,
we are given a graph $G$ and a non-negative integer $k$, and the question is 
whether $G$ contains a set of $k$ vertices whose closed neighborhood contains all the vertices of $G$.  
The connected variant of the problem,  {\sc Connected Dominating Set} (\tCDS) asks, given a graph $G$ and a non-negative 
integer $k$, whether $G$ contains a dominating set $D$ of at most $k$ vertices such that 
for every connected component $C$ of $G$, we have that $G[V(C)\cap D]$ is connected.  
This definition of \tCDS\ differs slightly from the established one where one just  demands that  
the subgraph induced by the dominating set  be connected. Our definition generalizes the established 
one  to  include disconnected graphs. 
A  considerable part of the algorithmic  study of these {\sf NP}-complete  problems has been focused on the  design of parameterized and kernelization algorithms. 
In general,  \tDS \,  is {\sf W}[2]-complete 
and therefore it cannot be solved by a parameterized algorithm, 
unless an unexpected collapse occurs in the 
Parameterized Complexity hierarchy (see~\cite{DowneyF98,FlumGrohebook,Niedermeierbook06}) and thus also does not admit a kernel.  
However, there are  interesting graph classes where  {\em fixed-parameter tractable} {\sf  (FPT)}  algorithms exist
for the \tDS \, problem. The project of  widening the families of graph classes, on which such algorithms exist, inspired a multitude of  ideas that made \tDS \,  the test bed for some of the most cutting-edge techniques of parameterized algorithm design. For example, the initial study of parameterized subexponential algorithms for \tDS \, on planar graphs \cite{AlberBFKN02,DemaineFHT05talg,FominT06} resulted in the creation of bidimensionality theory
characterizing a broad range of graph problems  that admit efficient approximation schemes, fixed-parameter algorithms or kernels on a broad range of graphs \cite{DemaineFHT05sub,DemaineH07-CJ,DornFLRS10,FominLRS10,F.V.Fomin:2010oq,FominLS12}.

  One of  the  first  results  on linear kernels is the celebrated work of Alber et al.  on  \tDS \,  on planar graphs \cite{AlberFN04}. This work  augmented significantly the interest in  proving polynomial (or preferably linear) 
 kernels for other parameterized problems.   
  The result of Alber et al.~\cite{AlberFN04}, see also  \cite{ChenFKX07}, has been extended to much more general graph classes like graphs of bounded genus \cite{H.Bodlaender:2009ng} and apex-minor-free graphs \cite{F.V.Fomin:2010oq}.
An important step in this direction was made by  Alon and Gutner \cite{AG08TechReport,Gutner09}  who obtained a kernel of size $\cO(k^{h})$ for  \tDS \,  on \Hmf\,  and \tHmf\,  graphs, where the constant $h$ depends on the excluded graph $H$. Later, Philip et al.~\cite{PhilipRS09} obtained a kernel of size $\cO(k^{h})$ on  $K_{i,j}$-free and $d$-degenerate graphs, where $h$ depends on $i,j$ and $d$ respectively.  In particular, for $d$-degenerate graphs, a subclass of $K_{i,j}$-free graphs,  the algorithm of   Philip et al.~\cite{PhilipRS09} produces a kernel of size
 $\cO(k^{d^2})$. Similarly, the sizes of   the kernels in~\cite{AG08TechReport,Gutner09,PhilipRS09} are bounded by  polynomials in $k$ with degrees depending on the size of the excluded minor $H$. 
  Alon and Gutner \cite{AG08TechReport} mentioned as a  challenging question  whether one can  characterize the families of graphs for which the dominating set problem admits a linear kernel, i.e. a kernel  of size $f(h)\cdot k $, where the function $f$ depends {\em exclusively} on the 
graph family. 
In this direction, there are already results for more restricted graph classes.
According to
the meta-algorithmic results on kernels introduced in~\cite{H.Bodlaender:2009ng},  \tDS \,  has a kernel 
of size $f(g)\cdot k$ on graphs of genus $g$. An alternative meta-algorithmic 
framework, based on bidimensionality theory \cite{DemaineFHT05sub}, was introduced in~\cite{F.V.Fomin:2010oq}, implying the existence of a kernel of size $f(H)\cdot k$ for \tDS \, on graphs excluding an { apex\footnote{An {\em apex} graph is a graph that can be made planar by the removal of a single vertex.}} graph $H$ as a minor. While apex-minor-free graphs form much more general class of graphs than  graphs of bounded genus, \Hmf\,  graphs  and \tHmf\, graphs form much larger classes than apex-minor-free graphs. For example, the class of graphs excluding $H=K_6$, the complete  graph on $6$ vertices, as a minor, contains all apex graphs. Alon and Gutner in ~\cite{AG08TechReport} and Gutner  in~\cite{Gutner09} posed as an open problem  
whether one can obtain a linear kernel for \tDS \,  on \Hmf\,  graphs.
Prior to our work, the only result on linear kernels for \tDS \, on graphs excluding a fixed graph $H$ as a topological minor, was the result of  
 Alon and Gutner  in~\cite{AG08TechReport}  for the  special case where $H=K_{3,h}$.
See Fig.~\ref{fig:graph_classes} for the relationship between these classes.

\begin{figure}[t]
\begin{center}
\includegraphics[scale=0.287]{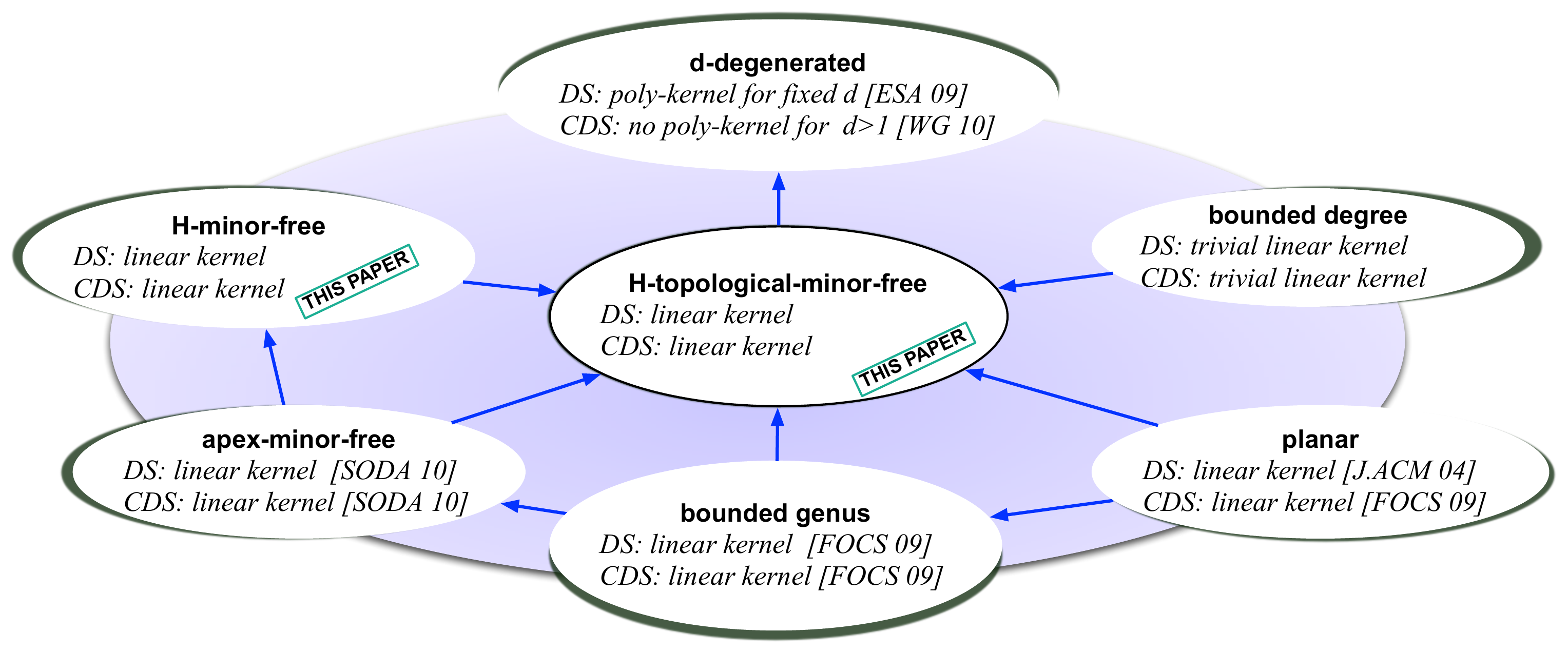}
\caption{Kernels for DS and CDS on classes of sparse graphs. Arrows represent inclusions of classes.
In the diagram, [J.ACM 04] refers to the paper of Albers et al. \cite{AlberFN04}, [FOCS 09]   to the paper of Bodlaender et al.~\cite{H.Bodlaender:2009ng},
[SODA 10] and [SODA 12] to the papers of Fomin et al. \cite{F.V.Fomin:2010oq,F.V.Fomin:2012},  [ESA 09] to the paper of Philip et al. \cite{PhilipRS09}, and  [WG 10] to the paper of Cygan et al. \cite{Cygan:2010bv}. }
\label{fig:graph_classes}
\end{center}
\end{figure}

 It is tempting to conjecture that similar improvements on kernel sizes are possible for more general graph classes like $d$-degenerate graphs. For example, for graphs of bounded vertex degree, a subclass of $d$-degenerate graphs, $\tDS$ has a trivial linear kernel. Unfortunately,  for $d$-degenerate graphs the existence of a linear kernel, or even a polynomial kernel with the exponent of the polynomial being independent of $d$, is very unlikely. 
   By the recent work of   Cygan et al. \cite{CyganGH12}, the kernelization algorithm of Philip et al.~\cite{PhilipRS09}  is essentially  tight---the existence of a kernel of size  
 $\cO(k^{(d-3)(d-1)-\varepsilon)})$ for  $\tDS$ on $d$-degenerate  graphs would imply that  {\sf coNP} is contained in {\sf NP/poly}.  
 
 In this work we show how to generalize the linearity of kernelization  for \tDS \, from 
 bounded-degree graphs and  apex minor free graphs to the class of graphs excluding a fixed graph $H$ as a topological minor.
  Moreover, a   modification of the ideas for \tDS \, kernelization can be used to obtain a linear kernel for \tCDS, which is usually a much more difficult problem to handle due to the connectivity constraint.  For example, \tCDS \, does not have a polynomial kernel on $2$-degenerate graphs  unless {\sf coNP} is in {\sf NP/poly}    \cite{Cygan:2010bv}. We must {\em emphasize} that our linear kernels are existential. That is, we just show the mere existence of polynomial time algorithms computing linear kernels. 
  
  The class of graphs excluding a fixed graph $H$ as a topological minor  is a wide class of graphs containing \Hmf\, graphs and graphs of constant vertex degrees. The existence of a linear kernel for \tDS \, on this class of graphs significantly extends and improves  previous works  
  \cite{AG08TechReport,F.V.Fomin:2012,Gutner09}.
  The extension of the results for planar graphs from \cite{AlberFN04} and apex-minor-free graphs from \cite{F.V.Fomin:2010oq} to the more general family of  \Hmf \, graphs requires several new ideas. Similar difficulties in 
  generlizing algorithmic techniques from apex-minor free to \Hmf \, graphs were observed in approximation   
\cite{Demaine:2009pd}  and parameterized algorithms  \cite{DemaineFHT05sub,DraganFG08}.  The basic idea behind kernelization algorithms on  apex-minor-free  graphs is the bidimensionality of \tDS. Roughly speaking, the treewidth of these graphs with dominating set of size $k$ is  $o(k)$.  In other words, 
 excluding an apex graph 
makes it possible to bound the tree-decomposability  of the input 
graph by a {\em sublinear} function of the size of a dominating set which is not the case for more general classes of \Hmf \, graphs or a family of graphs excluding a fixed graph $H$ as a topological minor.

A main ingredient of our kernelization algorithms are new reduction rules that allow us to obtain the desired kernels on \Hmf \, graphs.  This is an important step for our kernel on  the class of graphs excluding a fixed graph $H$ as a topological minor.   The main idea behind our algorithm is to identify and remove ``irrelevant" vertices  without changing the solution such that in the reduced graph one can select $\cO(k)$ vertices whose removal leaves protrusions, that is, subgraphs of constant treewidth separated from the remaining vertices by a constant number of vertices. If we are  able to 
obtain such a graph, we can use   the techniques from \cite{F.V.Fomin:2010oq} to construct the linear kernel. 
Roughly speaking, our rule to  identify  ``irrelevant"  vertices works as follows: 
we try specific vertex subsets of constant size, for each subset we try all ``feasible" scenarios how dominating sets can  interact with the subset, and find neighbours of theses 
subsets whose removal does not change the outcome of any feasible scenario. 
The main difference of this new reduction  rule in comparison to other rules for  \tDS \,   \cite{AlberFN04,ChenFKX07} is that instead of  reducing the size of the graph to $\cO(k)$, it reduces the treewidth of the graph to  $\cO(\sqrt{k})$. Thus idea-wise, it is closer  to 
 the ``irrelevant vertex'' approach  developed  by Robertson and Seymour  for 
  disjoint paths  and minor checking problems \cite{RobertsonS-XIII}. 
However, the significant difference with this technique is that in all applications of ``irrelevant vertex''   the bounds on the treewidth are exponential or even worse
\cite{KawarabayashiK08,Kawarabayashi:2010cs,Kobayashi:2009jt}. Moreover,  Adler et al.~\cite{Adler11} provide instances  of the  disjoint paths problem on planar graphs, for which the irrelevant vertex approach of  Robertson and Seymour produces graphs of treewidth $2^{\Omega(k)}$.  Our rule provides a reduced graph with   \emph{sublinear}    treewidth for $\tDS$.

The proof that after deletion of all irrelevant vertices the treewidth of the graph becomes sublinear is non-trivial.  For this proof we 
  need  the   theorem of Robertson and Seymour \cite{RobertsonS03}    on decomposing a graph into a set  of torsos connected via clique-sums. By making use of this theorem, we show that by applying the rule for all subsets of apex vertices of  each torso, it is possible to  reduce the treewidth  of each torso to $\cO(\sqrt{k})$. This implies that the treewidth of the reduced graph is also $\cO(\sqrt{k})$. However,  the number of torsos can be $\Omega(n)$ and the sublinear treewidth of the reduced graph  still does not bring us directly to the kernel. To  overcome this obstacle, we have to implement the irrelevant vertex rule  in a divide and conquer manner, and only after doing this can we guarantee that the reduced graph admits a linear kernel.  The idea of using divide and conquer in kernelization is our first conceptual contribution. 
  

%
  
    
The second main step of our kernelization algorithm for $\tDS$, on  the class of graphs excluding a fixed graph $H$ as a topological minor, is to design reduction rules for graphs of bounded degree. 
The ideas introduced for \Hmf \, graphs  can hardly  work on graphs of bounded degree, and hence on   graphs excluding a fixed graph $H$ as a topological minor. The reason is that the bound $o(k)$ on the treewidth of such graphs   would imply that  \tDS
\, is solvable in subexponential time on graphs of bounded degree, which in turn can be shown to contradict the  Exponential Time Hypothesis       
\cite{ImpagliazzoPZ01}. This is why the kernelization techniques developed for \Hmf \, graphs do not seem to be applicable directly in our case.

\paragraph{High level overview of the main  ideas.}
Our kernelization algorithm has two main phases. In the first phase we partition the input graph $G$ into subgraphs $C_0, C_1, \dots, C_\ell$, such that 
 $|C_0|=\cO(k)$; for every $i\geq 1$, the neighbourhood $N(C_i)\subseteq C_0$,  and $\sum_{1\leq i\leq \ell} |N(C_i)|=\cO(k)$.   In the second phase, we replace these graphs by smaller equivalent graphs.  Towards this, we treat 
  graphs $N[C_i]=C_i \cup N(C_i)$, $i\geq 1$,  as $t$-boundaried graphs with boundary $N(C_i)$. Our second  conceptual contribution is a  polynomial time algorithmic procedure for replacing a $t$-boundaried graph by an equivalent graph of size $\cO(|N(C_i)|)$.   Observe that as a result of such replacements, the size of the new graph is 
  $$\sum_{1\leq i\leq \ell} |\cO(N(C_i))| +|C_0|=\cO(k)$$ 
  and thus we obtain a linear kernel. 
  Kernelization techniques based on replacing a $t$-boundaried graph by an equivalent instance or, more specifically, protrusion replacement were used before in~\cite{H.Bodlaender:2009ng,F.V.Fomin:2010oq,FominLMPS11, abs-1207-0835}. At this point it is also important to mention earlier works done in~\cite{FellowsL89,ArnborgCPS93,BodlaendervA01a,Fluiter97,BodlaenderH98}  
  on protrusion replacement in the algorithmic setting on graphs of bounded treewidth. 
  The substantial differences  with our replacement 
  procedure  and the ones used before in the kernelization setting are the following. 
  \begin{itemize}
  \item In the protrusion replacement procedure it is assumed that the size of the boundary $t$ and the treewidth  of the replaced graph are constants. In our case neither the treewidth, nor the boundary size are  bounded. In particular, the boundary size  could be a {\em linear} function of $k$.
  \item In earlier protrusion replacements, the size of the equivalent replacing graph is bounded by some (non-elementary) function of $t$. In our case this is a {\sl linear} function of $t$.  
  \end{itemize} 
Our new  replacement procedure strongly exploits the fact that graphs $C_i$ possess 
 a set of desired properties allowing us to apply the irrelevant vertex technique explained above. 
  However, not every graph $G$ excluding some fixed graph   as a topological minor can be partitioned into graphs with the desired properties.  We show that, in this case, there is  another polynomial time procedure transforming $G$ into an equivalent graph, which in turn can be partitioned. The procedure is based on a generalized notion of protrusion, which is the third conceptual contribution of this paper. In the new notion of protrusion we relax  the requirement that protrusions are of bounded treewidth by the condition   that they  have a bounded size dominating set. Let us remark, that a similar notion of a generalized protrusion, bounded by the size of a certificate, can be used  for a variety of graph problems. We show that
either   a graph does not have the desired partition, or  it   contains a  sufficiently large generalized protrusion, which can be replaced by a smaller equivalent subgraph.  
The construction of the partitioning  is heavily based on the recent work of Grohe and Marx on the structure of  such graphs  \cite{GroheM12}.

   As a byproduct of our results we  obtain the first subexponential time  algorithms for {\sc Connected Dominating Set}, a deterministic algorithm  solving the problem on an $n$-vertex \Hmf \, graph in  time $2^{\cO(\sqrt{k})}+ n^{\cO(1)}$.  For {\sc Dominating Set} our results imply a significant 
simplification and refinement  of  a $2^{\cO(\sqrt{k})}n^{\cO(1)}$ algorithm on \Hmf \, graphs due to Demaine et al. \cite{DemaineFHT05sub}. 
Also our kernels can be used to obtain, subexponential, polynomial-space parameterized algorithms for these problems.    

\paragraph{Organization of the paper.}
The remaining part of this paper is organized as follows.
In Section~\ref{sec:defs_and_nots}, we provide definitions and  state known results used in the paper. In Section~\ref{sec:gens:protrs}, we introduce the new notion of ``generalized protrusions" and build a theory of  replacements for such protrusions. 
We provide a decomposition lemma in Section~\ref{sec:slicedecs}, which will be used  for kernelization algorithms. In Sections~\ref{sec:domset_kernel} and~\ref{sec:CDSkernel} we give the two main results of the paper, linear kernels for \tDS\ and \tCDS\ on  the class of graphs excluding a fixed graph $H$ as a topological minor. In Section~\ref{sec:concludes}  we conclude with questions for further research and give a  short overview of some of the developments which happened since the conference versions of this paper were published, including work on  kernelization of \tDS\ and \tCDS\ on graphs of bounded expansion and on nowhere-dense graphs.

\section{Preliminaries}
\label{sec:defs_and_nots}

In this section we give various definitions which we make use of in the paper. We refer to Diestel's book 
\cite{diestelbook} for standard definitions from Graph Theory.
 Let~$G$ be a graph with vertex set $V(G)$ and edge set $E(G)$.  A graph~$G' $ is a
 \emph{subgraph} of~$G$ if~$V(G') \subseteq V(G)$ and~$E(G') \subseteq E(G)$.
 For a subset $V'\subseteq V(G)$, the subgraph~$G'=G[V']$ of $G$ is called the   \emph{subgraph induced  by $V'$}  if~$E(G')
 = \{uv \in E(G) \mid u,v \in V'\}$.
  By $N_G(u)$
we denote the (open) neighborhood of $u$ in graph $G$, that is, the set of all vertices
adjacent to $u$ and by $N_G[u]=N_G(u)\cup \{u\}$. 
Similarly, for a subset $D \subseteq V$, we define $N_G[D]=\cup_{v\in D} N_G[v]$
and $N_G(D)=N_G[D]\setminus D$.  Given a set $S\subseteq V (G)$, we define $\partial_G(S)$ as the set of vertices in 
$S$ that have a neighbor in $V(G) \setminus S$.  We omit the subscripts when they are clear from the context. A subset of vertices $D$ is called a {\em dominating set} of $G$ if 
$N[D]=V(G)$. A subset of vertices $D$  is called a {\em connected dominating set} if it is a dominating set and for every connected component $C$ of $G$ we have that $G[D\cap C]$ is connected. Throughout the paper, given a graph $G$ and vertex subsets $Z$ and $S$, whenever we say that a subset $Z$ 
{\em dominates all but (everything but)} $S$ then we mean that $V(G)\setminus S \subseteq N[Z]$. Observe that a vertex of $S$ can also be dominated by the set $Z$.

\medskip

We denote by $K_h$ the complete graph on $h$ vertices. Also for a given graph $G$ and a vertex subset $S$ by $K[S]$ we mean a clique on 
the vertex set $S$. 
For an integer $r\geq 1$ and  vertex subsets $P,Q\subseteq V(G)$, we say that a subset $Q$ is \emph{$r$-dominated} by $P$, if for every $v\in Q$  there is $u\in P$ such that the distance between $u$ and $v$ is at most $r$. For $r=1$, we simply say that $Q$ is dominated by $P$.  We denote by  $N^r_G(P)$ the set of vertices $r$-dominated by $P$.  

Throughout this paper we use $\mar{\Bbb{Z}}$,  $\mar{\Bbb{Z}^{+}}$  and  $\mar{\Bbb{Z}^{-}}$  for the sets of  integers,  non-negative and 
non-positive integers respectively.   Finally,  we use $\mar{\Bbb{N}}$ for the set of positive integers.
\paragraph{Minors and Contractions.}
Given an edge  $e=xy$ of a graph $G$, the graph  $G/e$ is
obtained from  $G$ by contracting the edge $e$, that is,
the endpoints $x$  and $y$ are replaced by a new vertex $v_{xy}$
which  is  adjacent to the old neighbors of $x$ and $y$ (except
from $x$ and $y$).  A graph $H$ obtained by a sequence of
edge-contractions is said to be a \emph{contraction} of $G$.  We denote it by $H\leq_{c} G$.
A graph $H$ is a {\em minor} of a graph $G$ if $H$ is the contraction of some subgraph
of $G$ and we denote it by $H\leq_{m} G$.
We say that a graph $G$ is {\Hmf \,} when it does not
contain $H$ as a minor. We also say that a graph class ${\cal G}_H$
is {\Hmf \,} (or, excludes $H$ as a minor)  when
all its members are \Hmf.
 An \emph{apex graph} is a graph obtained from a planar graph $G$
by adding a vertex and making it adjacent to some of the vertices of $G$.
A graph class ${\cal G}_H$ is \emph{apex-minor-free} if ${\cal G}_H$
excludes a fixed apex graph $H$ as a minor.


A \emph{subdivision} of a graph $H$ is obtained by replacing each edge of $H$ by a non-trivial path. We say that $H$ is a \emph{topological minor} of $G$ if some subgraph of $G$ is isomorphic to a subdivision of $H$ and denote it by $H\preceq_T G$. 
A graph $G$ \emph{excludes a graph $H$ as a (topological) minor} if $H$ is not a (topological) minor of $G$.  For a graph $H$, 
by $\mathcal{C}_H$, we denote all graphs that exclude $H$ as  topological minors.


%

\paragraph{ Tree-Decompositions.} A \emph{tree-decomposition} of a graph $G$ is a pair $(M,\Psi)$ where $M$ 
is a rooted tree and $\Psi:V(M)\rightarrow 2^{V(G)}$, such that :

\begin{enumerate}
\item $\bigcup_{t\in V(M)}\Psi(t)=V(G)$.
\item For each edge $uv\in E(G)$, there is a $t\in V(M)$ such that both $u$ and $v$ belong to $\Psi(t)$.
\item For each $v\in V(G)$, the nodes in the set $\{t\in V(M)\mid v\in \Psi(t)\}$ form a subtree of $M$.
\end{enumerate}
If $M$ is a path then we call the pair  $(M,\Psi)$  as {\em path-decomposition}.

\noindent
The following notations are the same as that in \cite{GroheM12}. Given a tree-decomposition of a graph 
$G$, 
we define mappings $\sigma, \gamma :V(M)\rightarrow2^{V(G)}$  and $\kappa: E(M) \rightarrow 2^{V(G)}$. 
 For all $t\in V(M)$,\\
\begin{center}
$
\sigma(t) = \begin{cases}

  \emptyset & \text{if $t$ is the root of $M$} \\

  \Psi(t)\cap \Psi(s) & \text{if $s$ is the parent of $t$ in $M$} \\

\end{cases}
$\end{center}

$$\gamma(t)=\bigcup_{u \text { is a descendant of } t}\Psi(u)$$\\


For all $e=uv \in E(M)$, $\kappa(e)=\Psi(u)\cap \Psi(v)$. 

For a subgraph $M'$ of $M$ by $\Psi(M')$ we denote $\cup_{t\in V(M')}\Psi(t)$. 
\medskip
\noindent


Let $(M,\Psi)$ be a tree-decomposition of a graph $G$. The {\em width} of $(M,\Psi)$ is 
$$\min\Big\{\vert \Psi(t)\vert -1 \mid t\in V(M)\Big\},$$ and the {\em adhesion}  of the tree-decomposition is 
$$\max\Big\{\vert \sigma(t)\vert \mid t\in V(M)\Big\}.$$ 
We use $\tw(G)$ to denote the treewidth of the input graph. For every node $t\in V(M)$, the {\em torso} at $t$ is the graph 

\begin{center}
$\tau(t):= G[\Psi(t)] \cup E(K[\sigma(t)]) \cup \bigcup_{u\, \mathrm{ child}\, \mathrm{of }\, t} E(K[\sigma(u)])$.
\end{center}

We take the graph induced by $\Psi(t)$, turn $\sigma(t)$ into a clique, and make vertices $x,y$ adjacent if they appear together in the separator  of some child $u$ of $t$.

\paragraph{Parameterized graph problems.}
A parameterized graph problem  $\Pi$ is usually defined as a subset of $\Sigma^{*}\times \Bbb{Z}^{+}$
where, in each instance $(x,k)$ of $\Pi,$ $x$ encodes a graph and $k$ is the parameter (we denote by $\Bbb{Z}^{+}$ the set of all non-negative integers). In this paper we use an extension of this definition (also used by Bodlaender et al.~\cite{H.Bodlaender:2009ng}) that permits the parameter $k$ to be negative 
with the additional constraint that either all pairs with non-positive values of the parameter 
are in $\Pi$ or that no such pair is in $\Pi$. Formally, a parametrized problem $\Pi$
is a subset of $\Sigma^{*}\times \Bbb{Z}$ where for all $(x_{1},k_{1}),(x_{2},k_{2})\in\Sigma^{*}\times \Bbb{Z}$
with $k_{1},k_{2}<0$ it holds that $(x_{1},k_{1})\in\Pi$ if and only if  $(x_{2},k_{2})\in\Pi$.
This extended definition encompasses the traditional one and is needed for technical reasons  
(see Subsection~\ref{subsec:finiinteginde}).
In an instance of a parameterized problem $(x,k),$ the integer $k$ is called the parameter. Now we formally define the 
$\tDS$ and $\tCDS$ problems. 

\defparprob{$\tDS$}{An undirected graph $G$ and a positive integer $k$.}{$k$}{Does there exists $D\subseteq V(G)$ of size at most $k$ such that $N[D]=V(G)$?}

\medskip

\defparprob{$\tCDS$}{An undirected graph $G$ and a positive integer $k$.}{$k$}{Does there exists $D\subseteq V(G)$ of size at most $k$ such that $N[D]=V(G)$ and $G[D]$ is connected?}


\paragraph{Kernels   and Protrusions.}
A central notion in 
parameterized complexity is {\em fixed parameter tractability}, which means, 
for a given instance $(x,k),$ 
solvability in time $f(k)\cdot p(|x|),$ where $f$ is an arbitrary function of $k$ and $p$ is a polynomial function in the input size. 
The notion of {\em kernelization} is formally defined as follows.

\begin{definition}
A {\em{kernelization algorithm}}, or simply a {\em{kernel}}, for a parameterized problem $\Pi$ is an algorithm $\mathcal{A}$ that, given an instance $(x,k)$ of $\Pi$, works in polynomial-time and returns an equivalent instance $(x',k')$ of $\Pi$. Moreover, 
there exists a computable function $g(\cdot)$ such that whenever $(x',k')$ is the output for an instance $(x,k)$, then it holds that $|x'|+k'\leq g(k)$. If the upper bound $g(\cdot)$ is a polynomial (linear) function of the parameter, then we say that $\Pi$ admits a {\em{polynomial (linear) kernel}}.
\end{definition}

We often abuse the notation and call the output of a kernelization algorithm, the ``reduced'' equivalent instance, also a kernel.

 
 \begin{definition}
Given a graph $G$, we say that a set $X\subseteq V(G)$ is an {\em $r$-protrusion} of $G$ if 
  $\tw(G[X])\leq r$ and the number of vertices in $X$ with a neighbor in $V(G)\setminus X$ is at most $r$. 
  \end{definition}

\subsection{Known Decomposition Theorems} We start with the definition of nearly embeddable graphs. 
\begin{definition}[$h$-nearly embeddable graphs]
Let $\Sigma$ be a surface with boundary cycles $C_1, \dots,C_h$, i.e.\ each cycle
$C_i$ is the border of a disc in $\Sigma$. A graph $G$ is
{\em $h$-nearly embeddable} in $\Sigma$, if $G$ has a subset $X$ of size at most $h$,
called {\em apices}, such that there are (possibly empty) subgraphs
$G_0=(V_0, E_0),\dots,G_h=(V_h, E_h)$ of $G\setminus X$ such that
\begin{itemize}
\setlength{\itemsep}{-2pt}
\item $G \setminus X=G_0\cup\dots\cup G_h$,
\item $G_0$ is embeddable in $\Sigma $, we fix an embedding of $G_0$,
\item graphs $G_1,\dots,G_h$ (called \emph{vortices}) are pairwise disjoint,
\item for $1\leq i \leq h$, let $U_i:= \{u_{i_1},\dots,u_{i_{m_i}}\} = V_0 \cap V_i$,  $G_i$ has a path decomposition 
$(B_{ij},\Psi_{ij}),\ 1\leq j \leq m_i$, of width at most $h$ such that

\begin{itemize}
\item for $1\leq  i \leq h$ and for $1 \leq j \leq m_i$ we have $u_{i_j} \in B_{ij}$
\item for $1\leq i \leq h$, we have $V_0 \cap C_i = \{u_{i_1},\dots,u_{i_{m_i}}\} $ and the points $u_{i_1},\dots,u_{i_{m_i}}$ appear on $C_i$ in this order (either if we walk clockwise or anti-clockwise).
\end{itemize}
\end{itemize}
\end{definition}

%
%
%


\noindent
The decomposition theorem that we use extensively for our proofs is given in the next theorem. 
\begin{theorem}  [\cite{abs-1209-0129,GroheM12,RobertsonS03}]
\label{thm:structure theorem}  
For every graph $H$, there exists a constant $h$, 
depending only on the size of $H$, such that every graph $G$  with $H\not\preceq_T G$, there is a tree-decomposition 
$(M,\Psi)$ of adhesion at most $h$ such that for all $t\in V(M)$, one of the following conditions is satisfied: 
\begin{enumerate}
\item  $\tau(t)$  is $h$-nearly embedded in a surface $\Sigma$ in which $H$ cannot be embedded.
\item $\tau(t)$ has at most $h$ vertices of degree larger than $h$.
\end{enumerate}
Moreover, if $G$ is an \Hmf \,   then nodes of second type do not exist. 
\noindent
Furthermore, there is an algorithm that, given graphs $G$, $H$ on $n$ and $|V(H)|$ vertices, respectively, computes such a tree-decomposition in time 
$f(|V(H)|)n^{\cO(1)}$ for some computable function $f$,  and moreover computes an apex set $Z_t$ of size at most $h$  
for every bag of the first type.
\end{theorem}
One of the main consequence of Theorem~\ref{thm:structure theorem} we need for our purposes is  that (in the case when 
$G$ is \Hmf) for every 
$H$ there exist  constants $h$ and $h'$ such that  for every torso $L$  of the decomposition from 
Theorem~\ref{thm:structure theorem}, there exists a set of vertices $A \subseteq V(L)$ of size at most $h$, called apices, such 
that the graph obtained from $L$ after deleting the apices does not contain some  apex graph $H'$ of size 
$h'$ as a minor. See, e.g.~\cite[Theorem $13$]{Grohe03}. 

 Furthermore we can assume that in  $(M,\Psi)$, for any $x,y\in V(M)$, $\Psi(x)\not \subseteq \Psi(y)$. That is, no bag is contained in other. 
 See~\cite[Lemma 11.9]{FlumGrohebook} for the proof.  
 
 
 

%

\subsection{Known Approximation Algorithms}
Recall that by $\mathcal{C}_H$ we denote the class of graphs that exclude a fixed graph $H$ as a topological minor. 
In this subsection we state known polynomial-time constant factor approximation algorithms  for \tDS \, and \tCDS\ on $\mathcal{C}_H$. It is well known that graphs in 
$\mathcal{C}_H$ has bounded 
degeneracy. The following is known about the approximation of \tDS. 

%
\begin{lemma}[\cite{Dvorak13}]
\label{lemma:approximation}
Let $H$ be a   graph. Then there exists a constant $\eta(H)$ depending only on $|V(H)|$ such that 
{\sc Dominating Set}  admits  a $\eta(H)$-factor approximation algorithm on  $\mathcal{C}_H$. 
\end{lemma}

For \tCDS \,  we  need the following proposition attributed to \cite{Duchet82}.
\begin{proposition}
\label{lem:bb}
Let $G$ be a connected graph and let $Q$ be a dominating set of $G$ such that $G[Q]$ has at most $\rho$ connected components. Then 
there exists a set $Z\subseteq V(G)$ of size  at most $2\cdot (\rho-1)$  
 such that $Q\cup Z$ is a connected dominating set in $G$ and, given $Q$, we can find such a set $Z$ in polynomial time. 
\end{proposition}

Combining Lemma~\ref{lemma:approximation} and Proposition~\ref{lem:bb} we arrive at the following. 
\begin{lemma}
\label{lemma:approximationcds}
Let $H$ be a   graph and $\eta(H)$ the constant from  Lemma~\ref{lemma:approximation}. Then
\tCDS \,  admits  a $3\eta(H)$-factor approximation algorithm on  $\mathcal{C}_H$. 
\end{lemma}


\section{A New Algorithm for Protrusion Replacement}
\label{sec:gens:protrs}

In the next section we introduce the notion of a ``generalized protrusion''.  Recall that a protrusion in a graph  is a subgraph of constant treewidth which is separated from the rest of the graph by at most a constant number of vertices. In our variant of protrusions, instead of stipulating that the subgraph be of constant treewidth, we ask that it contains a  constant number of vertices from a solution. In this section we show that even if we have a generalized protrusion then we can find a replacement for it efficiently.  
We first start with some relevant definitions and concepts. 

\subsection{Boundaried Graphs} 
\label{subsec:boungrap}
Here we define the notion of {\em boundaried graphs} and various operations on them.
\begin{definition}{\rm [\bf Boundaried Graphs]}\label{def:boungraph}
A \term{boundaried graph} is a graph $G$ with a set $B\subseteq V(G)$ 
of  distinguished vertices and an injective labelling $\mar{\lambda}$ 
from $B$  to the set $\Bbb{Z}^{+}$. The set $B$ is called the \term{{\em boundary}} of $G$ and  the vertices in $B$  are called  {\em boundary vertices} or \term{{\em terminals}}. 
Given a boundaried graph $G,$ we denote its boundary by $\mar{\delta(G)},$
we denote its labelling by $\lambda_G$, 
and we define its {\em label set} by $\mar{\Lambda(G)}=\{\lambda_{G}(v)\mid v\in \delta(G)\}$.
Given a finite set $I\subseteq \Bbb{Z}^{+}$, we define 
$\mar{{\cal F}_{I}}$  to denote the class of all boundaried graphs whose label set is $I$. 
We also denote by $\mar{{\cal F}}$ the class of all boundaried graphs.
Finally we say that a boundaried graph is a {\em $t$-boundaried} graph if $\Lambda(G)\subseteq \{1,\ldots,t\}$.
\end{definition}

%

%

\begin{definition}{\rm [\bf Gluing by $\oplus$]} Let $G_1$ and $G_2$ be two  boundaried graphs. We denote by $G_1 \mar{\oplus} G_2$ the  graph 
(not boundaried) obtained by taking the disjoint union of $G_1$ and $G_2$ and identifying equally-labeled vertices of the boundaries of $G_{1}$ and $G_{2}.$ In $G_1 \oplus G_2$ there is an edge between two vertices if there is  an edge between them either in $G_1$ or in $G_2$, or both.  
\end{definition}

We remark that if $G_1$ has a label which is not present in $G_2$, or vice-versa, then in $G_1 \oplus G_2$ we just forget that label. 

\begin{definition} {\rm [\bf Gluing by $\oplus_\delta$]}
The {\em boundaried gluing operation} $\oplus_{\delta}$ is similar to the normal gluing operation, but results in a boundaried graph rather than a graph. Specifically $G_1 \oplus_\delta G_2$ results in a boundaried graph where the graph is $G = G_1 \oplus G_2$ and a vertex is in the boundary of $G$ if it was in the boundary of $G_1$ or  of $G_2$. Vertices in the boundary of $G$ keep their label from $G_1$ or $G_2$. 
\end{definition}


Let ${\cal G}$ be a class of (not boundaried)  graphs.
By slightly abusing notation we say that a boundaried graph {\em belongs to a graph class ${\cal G}$} if the underlying graph belongs to ${\cal G}.$

\begin{definition}{\rm [\bf Replacement]}\label{defn:replacement}
Let $G$ be a $t$-boundaried graph containing a set $X\subseteq V(G)$
such that $\partial_{G}(X)=\delta(G).$ Let $G_1$ be a $t$-boundaried graph. The result of {\em replacing $X$ with $G_1$} is the graph $G^{\star}\oplus G_{1},$
where $G^{\star}=G\setminus (X\setminus \partial (X))$ is treated as a $t$-boundaried graph with  $\delta(G^{\star})=\delta(G).$
\end{definition}

\subsection{Finite Integer Index}
\label{subsec:finiinteginde}
\begin{definition}{\rm [\bf Canonical equivalence on boundaried graphs.]}
Let $\Pi$ be a parameterized graph problem whose instances are pairs of the form $(G,k).$
 Given two boundaried graphs $G_1,G_2~\in {\cal F},$ we say that \term{$G_1\!\equiv _{\Pi}\! G_2$} if 
$\Lambda(G_{1})=\Lambda(G_{2})$
 and there exists a \term{{\em transposition constant}}
$c\in\Bbb{Z}$ such that 
\begin{eqnarray*}
\forall(F,k)\in {\cal F}\times \Bbb{Z} &&  (G_1 \oplus F, k) \in \Pi \Leftrightarrow (G_2 \oplus F, k+c) \in \Pi.\label{eq:fiidef}
\end{eqnarray*}
Here, $c$ is a function of the two graphs $G_1$ and $G_2$. 
\end{definition}
Note that  the relation $\equiv_{\Pi}$  is
an equivalence relation. Observe that $c$ could be negative in the above definition. This is the reason we allow the parameter in parameterized problem instances to take negative values.

Next  we define a notion of ``transposition-minimality'' for the members 
of  each equivalence class of $\equiv_{\Pi}.$

\begin{definition}{\rm [\bf Progressive representatives]}
\label{def:progrepr}
Let $\Pi$ be a parameterized graph problem whose instances are pairs of the form $(G,k)$
and let ${\cal C}$ be some equivalence class of $\equiv_{\Pi}$. We say that $J\in{\cal C}$ is a \term{{\em progressive 
representative}}
of ${\cal C}$ if for every $H\in{\cal C}$
there exists $c\in\Bbb{Z}^{-},$ such that 
\begin{eqnarray}
\forall(F,k)\in {\cal F}\times \Bbb{Z} \ \ \  (H \oplus F, k) \in \Pi \Leftrightarrow (J\oplus F, k+c) \in \Pi. \label{eq:progfii}
\end{eqnarray}
\end{definition}

The following lemma guarantees the existence of a progressive representative for each equivalence class of 
$\equiv_{\Pi}$. 

\begin{lemma}[\cite{H.Bodlaender:2009ng}]
\label{lem:existprog}
Let $\Pi$ be a parameterized graph problem whose instances are pairs of the form $(G,k)$.
 Then each  equivalence class of $\equiv_{\Pi}$ has a progressive representative.
\end{lemma}

Notice that two  boundaried graphs with different label sets belong to 
different equivalence classes of $\equiv_{\Pi}.$ Hence for every equivalence 
class ${\cal C}$ of $\equiv_{\Pi}$ there exists some finite set $I\subseteq\Bbb{Z}^{+}$ such that 
${\cal C}\subseteq  {\cal F}_{I}$. We are now in position  to give the following definition.

\begin{definition}{\rm [\bf Finite Integer Index]}
\label{def:deffii}
A parameterized graph problem $\Pi$ whose instances are pairs of the form $(G,k)$
has {\em Finite Integer Index} (or  is \term{{\em FII}}), if and only if for every finite $I\subseteq \Bbb{Z}^+,$
the number of equivalence classes of  $\equiv_{\Pi}$ that are subsets of ${\cal F}_{I}$
is finite. For each $I\subseteq \Bbb{Z}^{+},$ we define $\mar{{\cal S}_I}$ to be
a set containing exactly one progressive representative of each equivalence class of $\equiv_{\Pi}$
that is a subset of ${\cal F}_{ I}$. We also define $\mar{{\cal S}_{\subseteq I}}=\bigcup_{I'\subseteq I} \mar{{\cal S}_{I'}} $. 
\end{definition}

\subsection{Replacement lemma}
%

We first define a notion of monotonicity for parameterized problems. 

\begin{definition}
We say that a parameterized graph problem $\Pi$ is {\em positive monotone} if for every graph $G$ 
there exists a unique $\ell \in \Bbb{N}$ such that for all $\ell'\in \mathbb{N}$ and $\ell' \geq \ell$, $(G,\ell')\in \Pi$ and for all 
$\ell'\in \mathbb{N}$ and $\ell' < \ell$, $(G,\ell')\notin \Pi$.  A parameterized graph problem $\Pi$ is {\em negative monotone} if for every graph $G$ 
there exists a unique $\ell \in \Bbb{N}$ such that for all $\ell'\in \mathbb{N}$ and $\ell' \geq \ell$, $(G,\ell')\notin \Pi$ and for all 
$\ell'\in \mathbb{N}$ and $\ell' < \ell$, $(G,\ell')\in \Pi$. $\Pi$ is monotone if it is either positive monotone or negative monotone.  
We denote the integer $\ell$ by {\sc Threshold($G,\Pi$)} (in short  {\sc Thr($G,\Pi$)}). 
\end{definition}

We first give an intuition for the next definition.  We are considering monotone functions and thus for every graph $G$ 
there is an integer $k$ where the answer flips. However, for our purpose we need a corresponding notion for 
boundaried graphs.   If we think of the representatives as some ``small perturbation'' , then it is the max threshold over all small perturbations (``adding a representative = small perturbation''). This leads to the following definition. 

\begin{definition}
Let $\Pi$ be a monotone parameterized graph problem that has FII. Let  ${\cal S}_t$  be
a set containing exactly one progressive representative of each equivalence class of $\equiv_{\Pi}$ that is a subset of 
${\cal F}_{I}$, where $I=\{1,\ldots,t\}$.  
For a $t$-boundaried graph $G$, we define   
\[\iota(G)= \max_{G'\in {\cal S}_t}  \mbox{{\sc Thr($G\oplus G',\Pi$)}}. \]
\end{definition}


The next lemma says the following. Suppose we are dealing with some FII problem and we are given a boundaried graph with constant size boundary.  We know it has some constant size representative and we want to find this representative. 
Now in general finding a representative for a boundaried graph is more difficult than solving the corresponding problem
The next lemma says basically that if ``OPT'' of a boundaried graph is low, then we can efficiently find its representative. 
Here by ``OPT''  we mean $\iota(G)$, which is a robust version of the threshold function (under adding a representative). 
And by efficiently we mean as efficiently as solving the problem on normal (unboundaried) graphs if we know that ``OPT'' is low. Specifically, the main result of this section is as follows. 

\begin{lemma}
\label{lem:red2finiteindex}
Let $\Pi$ be a monotone parameterized graph problem that has FII. Furthermore, let $\cal A$ be an 
algorithm for $\Pi$ that, given a pair $(G,k)$, decides whether it is in $\Pi$ in time $f(|V(G)|,k)$. 
Then for every $t\in\Bbb{N},$ there exists a $ \xi_t \in\Bbb{Z}^{+}$ (depending on $\Pi$ and $t$), and 
an algorithm that, given a $t$-boundaried graph $G$  with $|V(G)|>\xi_t,$ outputs, in  $\cO(\iota(G)(f(|V(G)|+\xi_t,\iota(G)))$ steps,
a $t$-boundaried graph $G^*$  such that $G\equiv_{\Pi}G^*$ and  $|V(G^*)| < \xi_t$. Moreover we can compute the translation 
constant  $c$ from $G$ to $G^*$ in the same time.
\end{lemma}

\begin{proof}
We give prove the claim for positive monotone problems $\Pi$; the proof for negative monotone problems is identical. 
 Let  ${\cal S}_t$  be
a set containing exactly one progressive representative of each equivalence class of $\equiv_{\Pi}$ that is a subset of 
${\cal F}_{I}$, where $I=\{1,\ldots,t\}$, and  let  $\xi_t=\max_{Y\in {\cal S}_{t}}|V(Y)|.$ The set  ${\cal S}_{t}$ is hardwired in the description of the algorithm. 
 Let $Y_1,\ldots,Y_\rho$ be the set of progressive representatives in ${\cal S}_{t}$. Let ${\cal F}_{t}={\cal F}_{I}$. Our objective is to find 
  a representative $Y_\ell$  for  $G$ such that 
 \begin{eqnarray}
\forall (F,k)\in {\cal F}_{t}
\times \Bbb{Z} & &   (G \oplus F, k) \in \Pi  \Leftrightarrow    (Y_\ell \oplus F, k-\vartheta(X,Y_\ell)) \in \Pi. 
\label{eq:progresivereplacement}
\end{eqnarray}
Here, $\vartheta(X,Y_\ell)$ is a constant  that depends on $G$ and $Y_\ell$.  Towards this   
we define the following matrix for the set of representatives. Let 
$$A[Y_i, Y_j]=  \mbox{{\sc Thr($Y_i\oplus Y_j,\Pi$)}}$$
The matrix $A$ has constant size and is also hardwired in the description of the algorithm. 

Now given $G$ we find its representative as follows. 
\begin{itemize}
\item Compute the following row vector ${\cal X}=[ \mbox{{\sc Thr($G\oplus Y_1,\Pi$)}}, \ldots ,  
\mbox{{\sc Thr($G\oplus Y_\rho,\Pi$)}})]$. For each $Y_i$ we decide whether $(G\oplus Y_i,k)\in \Pi$ using the assumed algorithm for deciding 
$\Pi$,  letting $k$ increase from $1$ until the first time $(G\oplus Y_i,k)\in \Pi$. Since $\Pi$ is positive monotone this will happen for some 
$k\leq \iota(G)$. Thus the total time to compute the vector ${\cal X}$ is $\cO(\iota(G)(f(|V(G)|+\xi_t,\iota(G)))$. 

\item Find a translate row in the matrix $A(\Pi)$. That is, find an integer $n_o$ and a representative 
$Y_\ell$ such that  
\begin{eqnarray*}
[ \mbox{{\sc Thr($G\oplus Y_1,\Pi$)}},  \mbox{{\sc Thr($G\oplus Y_2,\Pi$)}}, \ldots ,  
\mbox{{\sc Thr($G\oplus Y_\rho,\Pi$)}}] \\
=[ \mbox{{\sc Thr($Y_\ell\oplus Y_1,\Pi$)}}+n_0,  \mbox{{\sc Thr($Y_\ell\oplus Y_2,\Pi$)}}+n_0, \ldots ,  
\mbox{{\sc Thr($Y_\ell\oplus Y_\rho,\Pi$)}}+n_0]
\end{eqnarray*}
Such a row must exist since ${\cal S}_t$ is  a set of representatives for $\Pi$; the representative $Y_\ell$ for the equivalence class to which $G$ belongs, satisfies the condition.  
\item Set $Y_\ell$ to be $G^*$ and the translation constant to be $-n_0$.
\end{itemize}
From here it easily follows that $G\equiv_{\Pi}G^*$. This completes the proof.  
\end{proof}
 We remark that the algorithm whose existence is guaranteed by the Lemma~\ref{lem:red2finiteindex} assumes that the set  ${\cal S}_{t}$ of representatives  are hardwired in the algorithm.  In its full generality we currently don't known of a procedure that for problems having FII outputs such a representative set. The application of Lemma~\ref{lem:red2finiteindex}  makes our kernelization algorithm non-constructive.

\section{Generalized  Protrusions and Slice-Decomposition}\label{sec:slicedecs}
%

In this section our objective is to show that in polynomial time we can  
partition the graph $G$  into parts that satisfy certain properties.  
 To obtain  our decomposition we need to use a more general  notion of 
protrusion. Recall that a protrusion in a graph  is a subgraph of constant treewidth which is separated from the rest of the graph by at most a constant number of vertices. In our variant of protrusions, instead of stipulating that the subgraph be of constant {\em treewidth}, we ask that it contains a {\em constant number of vertices from a solution}. More precisely, we need the following kind of protrusions. 
\begin{definition}{\rm [\bf $r$-{\sc DS}-protrusion]} 
 Given a graph $G$, we say that a set $X\subseteq V(G)$ is an {\em $r$-{\sc DS}-protrusion} of $G$ if 
   the number of vertices in $X$ with a neighbor in $V(G)\setminus X$ is at most $r$ and there exists a 
   subset $S\subseteq X$ of size at most $r$ such that  $S$ is a dominating set of $G[X]$. 
\end{definition}

%

The notion of a $r$-{\sc DS}-protrusion $X$ differs from a protrusion in the following way. In a 
protrusion  $\tw(X)$ is at most $r$, while in a $r$-{\sc DS}-protrusion we require that the dominating set of the graph induced by  $X$ is small. In the case of a $r$-{\sc DS}-protrusion, the treewidth could be unbounded. One can similarly define the notion of a $r$-$\Pi$-protrusion for other graph problems $\Pi$. Next we define a 
$r$-{\sc CDS}-protrusion. 

\begin{definition}{\rm [\bf $r$-{\sc CDS}-protrusion]} 
 Given a graph $G$, we say that a set $X\subseteq V(G)$ is an {\em $r$-{\sc CDS}-protrusion} of $G$ if 
   the number of vertices in $X$ with a neighbor in $V(G)\setminus X$ is at most $r$ and there exists a 
   subset $S\subseteq X$ of size at most $r$ such that  for every connected component $C$ of $G[X]$ we have that 
    $S\cap C$ is connected and is a dominating set for $C$. 
    \end{definition}

A natural question is what can we do if we get a large $r$-{\sc DS}-protrusion (or $r$-{\sc CDS}-protrusion).  The next lemma shows that in that case we can replace it with an equivalent small graph. 
We will also need the following. Let ${\cal G}$ be a graph class.  Given a  parameterized graph problem  $\Pi$ and a graph class ${\cal G},$ 
we denote by $\Pi\doublecap {\cal G}$ 
the problem obtained by 
removing from $\Pi$ all instances that 
encode graphs that do not belong to ${\cal G}.$  Our result is as follows.



\begin{lemma}
\label{lem:fiidomset}
Let $H$ be a fixed graph. For  every $t\in\Bbb{Z}^{+},$ there exist a $\xi_t\in\Bbb{Z}^{+}$ (depending on \tDS \, (\tCDS), $t$ and $H$),  and 
an algorithm $\cal A$ such that given a $t$-{\sc DS}-protrusion $X$ ($t$-{\sc CDS}-protrusion) with boundary  
$\partial(X)$, $|V(X)|>\xi_t,$ and $H\not\preceq_T X$, 
%
$\cal A$ outputs in   $\cO(|V(X)|)$ time ($|V(X)|^{\cO(1)})$ time),  
a $t$-boundaried graph  $X'$ such that  $H\not\preceq_T X'$ ($H\not\leq_m X'$) and  $X\equiv_{\tDS}X'$ ($X\equiv_{\tCDS}X'$) and $|V(X')|\leq \xi_t$. Moreover in the same time 
we can also find the translation constant $c$ from $X$  to $X'$.

\end{lemma}

\begin{proof}
Let $\cal G$ be the class of  graphs that excludes $H$ as a topological minor. 
For every $t\in\Bbb{Z}^{+}$ let $\xi_t$ be the constant as defined in Lemma~\ref{lem:red2finiteindex}. It is known 
that  \tDS$\doublecap {\cal G}$ (\tCDS $\doublecap {\cal G}$) is FII\cal~\cite{H.Bodlaender:2009ng} and monotone (see~\cite[Lemmas 7.3 and 7.4]{H.Bodlaender:2009ng}). Furthermore, \tDS \, and \tCDS \, can be solved in time  $\cO((hk)^{hk} n)$~\cite[Theorem 4]{AlonG09} 
and  $\cO(k^{\cO(h^2)k} n^{\cO(1)})$~\cite[Theorem 1]{GolovachV08}  respectively.  Here, $h=|V(H)|$ and $k$ is the parameter in the definitions of  \tDS \, and \tCDS. We  use these algorithms in Lemma~\ref{lem:red2finiteindex} with the parameter value being $r$. That is, $k:=r$. 
Thus, if $|V(X)|>\xi_t$  then 
by  Lemma~\ref{lem:red2finiteindex} in time  $\cO(|V(X)|)$  ($|V(X)|^{\cO(1)}$), we can obtain a 
$t$-boundaried graph $X'$  such that $X\equiv_{\tDS}X'$ ($X\equiv_{\tCDS}X'$),  $|V(X')| < \xi_t$ and 
$H\not\preceq_T X'$.   The last assertion that  $H\not\preceq_T X'$ follows from the fact that  \tDS$\doublecap {\cal G}$ is FII and thus all the graphs in the set of representatives with respect to $\equiv_{\tDS}$ belong to ${\cal G}$. 
Moreover, in the same time 
we can also find the translation constant $c$ from $X$  to $X'$ as done in Lemma~\ref{lem:red2finiteindex}. 

Let ${\cal G}^\star$  be the class of  graphs that excludes a fixed graph $H$ as a minor. It is known 
that  \tDS$\doublecap {\cal G}^\star$ (\tCDS $\doublecap {\cal G}^\star$) is FII\cal~\cite{H.Bodlaender:2009ng} and monotone. Thus, as in the case of $\cal G$, we can  obtain a 
$t$-boundaried graph $X'$  such that $X\equiv_{\tDS}X'$ ($X\equiv_{\tCDS}X'$),  $|V(X')| < \xi_t$ and 
$H\not\leq_m X'$.
\end{proof}




%
\begin{center}
\fbox{
\parbox{.95\textwidth}{
   Throughout this section we work on a graph $G$  that excludes a fixed graph $H$ as a topological minor. Here, $h$ will denote $|V(H)|$. 
   
   Furthermore, we assume that we have a (connected) dominating set $D$ such that the size of $D$ is at most $\eta(H)$-factor away ($3\eta(H)$-factor away) from the size of an optimal (connected) dominating set of $G$, obtained by applying Lemma~\ref{lemma:approximation}  (Lemma~\ref{lemma:approximationcds}) on the input graph $G$.   }
}
  \end{center}

Let $(M,\Psi)$ be a tree-decomposition of a graph $G$. For a subtree $M_i$ of $M$, we define ${\cal E}(M_i)$ as the set of edges in $M$ such that it has exactly one endpoint in $V(M_i)$. Furthermore we define $R_i^+= \Psi(M_i)$ and 
\begin{center}
$\tau(M_i):=  G[R_i^+] \cup \bigcup_{e\in {\cal E}(M_i)} E(K[\kappa(e)])$.
\end{center}
In plain words, $R_i^+$ denotes the union of bags corresponding to the nodes in $M_i$ and $\tau(M_i)$ is the graph induced on  $R_i^+$ with ``external adhesions'' being torsoed. 

Our main objective in this section is to obtain the following $(\alpha,\beta)${\em-slice decomposition} for $\alpha=\beta=\cO(k)$. 

\begin{definition}{\rm [{\bf $(\alpha,\beta)$-slice decomposition]}}
Let $H$ be a fixed graph and let $G$ be a graph with $H\not\preceq_T G$.  Let  $(M,\Psi)$ be the 
tree-decomposition given by Theorem~\ref{thm:structure theorem}. An $(\alpha,\beta)${\em-slice decomposition} of a graph $G$
is a collection $\cal P$ of pairwise vertex disjoint  subtrees $\{M_1,\ldots,M_\alpha\}$  of $M$ such that the 
following hold. 

\begin{itemize}
\item $\bigcup_{1\leq i \leq \alpha}  \Psi(M_i) = \bigcup_{1\leq i \leq \alpha} \left(\bigcup_{t\in V(M_i)} \Psi(t) \right) =V(G)$
\item There exists a graph $H^*$ whose size only depends on $h$, such that each $\tau(M_i)$ is either $H^*$-minor-free
or  has at most $h$ vertices of degree at least $h$.

\item $$\sum_{i=1}^{\alpha} \left(\sum_{e\in {\cal E}(M_i)} |\kappa(e)|\right) \leq \beta.$$

\end{itemize}
We refer to the sets $R^{+}_{i},$ $i\in\{1,\ldots,\alpha \},$  as the {\em slices} of ${\cal P}.$
\end{definition}

Essentially, the slice decomposition allows us to  partition the input graph $G$ into subgraphs $C_0, C_1, \dots, C_\ell$, such that 
 $|C_0|=\cO(k)$; for every $i\geq 1$, the neighbourhood $N(C_i)\subseteq C_0$,  and $\sum_{1\leq i\leq \ell} |N(C_i)|=\cO(k)$. To see this consider an instance $(G,k)$ of $\tDS$, where $G$ excludes a fixed graph $H$ as a topological minor.  Now obtain an   $(\alpha,\beta)$-slice decomposition for $\alpha=\beta=\cO(k)$ for $G$. We take 
 $$C_0 = \bigcup_{i=1}^{\alpha} \left(\cup_{e\in {\cal E}(M_i)} \kappa(e)\right),$$
and $C_i= \Psi(M_i) \setminus C_0$. One can easily verify that this partition of $V(G)$ satisfies the stated properties.  
 This is the decomposition we were talking about in the introduction. 
 
 Now we give a definitions that is useful in our procedure to find  the slice decomposition. 
\begin{definition}
Let $(M,\Psi)$ be the tree-decomposition of a graph $G$ given by  Theorem~\ref{thm:structure theorem}. For a subset 
$Q\subseteq V(G)$ and a subtree $M'$ of $M$ we define $\mu(M',Q)=|\Psi(M')\cap Q|$.  
\end{definition}

Let $(M,\Psi)$ be the tree-decomposition of a graph $G$ given by  Theorem~\ref{thm:structure theorem}.   
If we delete an edge $e=uv\in E(M)$ from the tree $M$ then 
we get two trees. We call the {\em trees as $M_u$ and $M_v$ based on whether they contain $u$ or $v$. } 

\begin{definition}
Let $(M,\Psi)$ be the tree-decomposition of a graph $G$ given by  Theorem~\ref{thm:structure theorem} and $D$ be the assumed dominating (connected) set of $G$.  We call a tree edge $e=uv\in E(M)$ 
{\em heavy} if $\mu(M_u,D)\geq h+1$ and $\mu(M_v,D)\geq h+1$. We use $\cal F$ to denote the set of heavy edges.  
\end{definition}



Our main lemma in this section shows that in polynomial time we can find an $(\cO(k),\cO(k))$-slice decomposition or 
a large $r$-{\sc DS}-protrusion (or $r$-{\sc CDS}-protrusion) or a large protrusion. In the latter cases we can apply either 
Lemma~\ref{lem:fiidomset} or a similar lemma developed in~\cite[Lemma~7]{H.Bodlaender:2009ng} for protrusions and reduce the graph.  

%

Before we prove the main result of this section, we prove some combinatorial properties of the set $\cal F$.  Given $\cal F$, by {\em subgraph of $M$ formed by the edges in $\cal F$} we mean a subgraph of $M$ whose vertex set consists of the end points of edges in $\cal F$ and the edge set is $\cal F$. 

\begin{lemma}
Let $M^*$ be the subgraph of $M$ formed by the edges in $\cal F$. Then $M^*$ is a subtree of $M$.
\end{lemma}
\begin{proof}
Clearly, $M^*$ is a forest as it is a subgraph of $M$.  To complete the proof we need to show that it is connected. We prove 
this using contradiction.  Suppose $M^*$ is a forest and  $M_i^*$ and $M_j^*$, $i\neq j$, are two maximal subtrees in $M^*$. Then we know that there exists a path $P$ in $M$ 
such that the first and the last edges are heavy and the path $P$ contains a light edge. 
Furthermore, we can assume that the first edge, say
$u_iv_i$, belongs to $M_i^*$ and the last edge, say  $u_jv_j$ belongs to $M_j^*$. Let a light edge on the path be $xy$. Now when we delete the 
edge $xy$ from $M$ we get two trees $M_x$ and $M_y$. We know that either $M_i^*\subseteq M_x$ and $M_j^*\subseteq M_y$ or vice versa. 
Suppose $M_i^*\subseteq M_x$ and $M_j^*\subseteq M_y$. Since $M_i^*$ contains the heavy edge $u_iv_i$ we have that 
$\mu(M_x,D)\geq h+1$. Similarly we can show that $\mu(M_y,D)\geq h+1$. This shows that $xy$ is a heavy edge, contradicting that $xy$ is light. One can similarly argue that $xy$ is a heavy edge  when  $M_i^*\subseteq M_y$ and $M_j^*\subseteq M_x$. This contradicts our assumption that $M^*$ is not a subtree of $M$. This completes the proof of the lemma.
\end{proof}

For our next proof we first give some well known observations about trees. Given a tree $T$, 
we call a node {\em leaf}, {\em link} or {\em branch} if its degree 
in $T$ is $1$, $2$ or $\geq 3$ respectively. Let
$S_{\geq 3 }(T)$ be the set of branch nodes, $S_{2 }(T)$
be the set of link nodes and $L(T)$ be the set of leaves in the
tree $T$. Let $\mathscr{P}_2(T)$ be the set of maximal paths
consisting entirely of link nodes.

\begin{fact}
\label{fact:simplecounta}
 $|S_{\geq 3 }(T)| \leq |L(T)|-1$.
\end{fact}
 
\begin{fact}
\label{fact:simplecountb}
 $|\mathscr{P}_2(T)| \leq 2 |L(T)|-1$.
\end{fact}
\begin{proof}
Root the tree at an arbitrary node of degree at least $3$. If there is no node of degree $3$ or more in $T$ 
then we know that $T$ is a path  and the assertion follows. Consider $T[S_2]$ which is the disjoint union of paths
$P\in \mathscr{P}_2(T)$. With every path $P\in \mathscr{P}_2(T)$, we
associate the unique child in $T$ of the last node of this path (furtherest from the root)  in
$T$. Observe that this association is injective and the associated
node is either a leaf or a branch node. Hence 
$$|\mathscr{P}_2(T)| \leq |L(T)|+ |S_{\geq 3 }(T)|\leq  2 |L(T)|-1$$
follows from Fact $1$.
\end{proof}

\begin{lemma}
\label{lemma:boundingleavesandpaths}
Let $M^*$ be the subgraph formed by the edges in $\cal F$. 
If $(G,k)$ is a yes instance of \tDS \, (\tCDS) then (a) $|L(M^*)| \leq \eta(H) k$;  
(b) $|S_{\geq 3}(M^*)|\leq \eta(H) k-1$; and 
(c)  $|\mathscr{P}_2(M^*)|\leq 2 \eta(H) k-1$.  
Here $\eta(H)$ is the factor of approximation in Lemma~\ref{lemma:approximation} (Lemma~\ref{lemma:approximationcds}). 
\end{lemma}
\begin{proof}
Root the tree at an arbitrary node $r$ of degree at least $3$ in $M^*$. If there is no node of degree $3$ 
or more in $M^*$ then we know that  $M^*$ is a path,  and the proof follows. 
We call a pair of nodes $u$ and $v$ {\em siblings} if they do not 
belong to the same path from the root $r$ in $M^*$. Observe that all the leaves of $M^*$ are siblings. 

 Let $D$ be an approximate solution to $\tDS$ ($\tCDS$) returned by applying Lemma~\ref{lemma:approximation} (Lemma~\ref{lemma:approximationcds}) on $G$. Since  $(G,k)$ is a yes instance we have that $|D|\leq  \eta(H) k$. 
Let $w_1,\ldots,w_\ell$ be the leaves of $M^*$ and let $e_1,\ldots,e_\ell$ be the  edges in $M^*$ incident 
with   $w_1,\ldots,w_\ell$, respectively.  To prove our first statement we will show that for every $i$, we have a vertex 
 $q_i\in D$  such that $q_i\in \gamma(w_i)$ and for  all $j\neq i$, $q_i\notin \gamma(w_j)$. This will establish an 
 injection from the set of leaves to the dominating set $D$ and thus the bound will follow. Towards this observe that 
 since the edge $e_i$ is heavy,  we have that $|\gamma(w_i)\cap D|\geq h+1$.  
 Furthermore, for every pair of 
 vertices $w_i,w_j\in L(M^*)$, $w_i\neq w_j$, we have that 
 $|\gamma(w_i)\cap \gamma(w_j)|\leq h$. The last assertion  follows from the fact that for a pair of siblings $w_i$ and $w_j$ 
  the only vertices that  can be in the intersection of  $\gamma(w_i)$ and $\gamma(w_j)$ must belong to both $\sigma(w_i)$ and 
  $\sigma(w_j)$. However, the size of any $\sigma(w_i)$ is upper bounded by $h$. This together with the fact that 
  $|\gamma(w_i)\cap D|\geq h+1$ implies that for every $i$, we have a vertex $q_i\in D$  such that 
 $q_i\in \gamma(w_i)$ and for  all $j\neq i$, $q_i\notin \gamma(w_j)$.  This implies that $|L(M^*)|\leq |D|$. However since $(G,k)$ is a yes instance to \tDS \, we have that $|D|\leq \eta(H)k$. This completes the proof of part (a)  of the lemma. 
Proofs for part (b) and part (c) of the lemma follow from Facts~\ref{fact:simplecounta} and \ref{fact:simplecountb}. 
\end{proof}

Before we prove our next lemma we show a lemma about dominating sets of subgraphs of $G$.
\begin{lemma}
\label{lem:smalldominatingset}
Let $H$ be a fixed graph and let $G$ be a graph with $H\not\preceq_T G$.  
Let $(M,\Psi)$ be the tree-decomposition of $G$ given by 
Theorem~\ref{thm:structure theorem} and let $D$ be a dominating set of $G$. If $M'$ is  a subtree of $M$,  then 
$$(D \cap \Psi(M') )\mathop{\cup}_{e\in {\cal E}(M')}\kappa(e)$$ is 
a dominating set for $G[\Psi(M')]$. 
\end{lemma}
\begin{proof}
The proof follows from the fact that $D \cap \Psi(M')$ dominates all the vertices in $\Psi(M')$ except possibly the ones that have 
neighbors in $V(G)\setminus (\mathop{\cup}_{e\in {\cal E}(M')}\kappa(e))$. Thus, $$(D \cap \Psi(M') )\mathop{\cup}_{e\in {\cal E}(M')}\kappa(e)$$ is 
a dominating set for $G[\Psi(M')]$. 
\end{proof}

Let $P_1, \ldots,P_\ell$ be the  paths in $ \mathscr{P}_2(M^*)$. We use $s_i$ and $t_i$ to denote the first  and the last vertices, respectively,  
of the path $P_i$. Since $P_i$ is a path consisting of link vertices,  we have that $s_i$ and $t_i$ have unique 
neighbors $s^*_i$ and $t^*_i$ respectively in $M^*$.  Observe that since $M^*$ is a subtree of $M$, we have that for every 
$i$, $P_i$ is also a path in $M$. 
If we delete the edges $s^*_is_i$ and $t^*_it_i$ from the tree $M$, then there is a subtree of $M$ that contains the path $P_i$; we call this subtree 
$M(P_i)$.  For any two vertices $a$ and $b$ on the path $P_i$ we use $P_i(a,b)$ to denote the subpath between $a$ and $b$ in $P_i$. 
Furthermore for  any subpath  $P_i(a,b)$, if we delete the edges incident to $a$  and $b$ on $P_i$ and not present in  $P_i(a,b)$ from the tree $M$, then there is a subtree of $M$ that contains the path $P_i(a,b)$; we call this subtree  $M(P_i(a,b))$.   

Now we recall the definition of $\xi_t$. Let $\Pi$ be a  monotone parameterized graph problem  that is FII. Then for every $t\in\Bbb{N},$ there exists a $ \xi_t \in\Bbb{Z}^{+}$ (depending on $\Pi$ and $t$), such that, given a $t$-boundaried graph $G$  with $|V(G)|>\xi_t,$  there exists a $t$-boundaried graph $G^*$  such that $G\equiv_{\Pi}G^*$ and  $|V(G^*)| < \xi_t$. 
In the next lemma we show that if any of the paths is ``too long'' then using a simple application of pigeonhole principle we can get a $2h$-{\sc DS}-protrusion. We use $|P_i|$ to denote the number of vertices in the path $P_i$.  

\begin{lemma}
\label{lem:boundingpaths}
Let $(G,k)$ be an instance of  \tDS\ (\tCDS) and let $P_1, \ldots,P_\ell$ be the  paths in $ \mathscr{P}_2(M^*)$. Further, let $D$ be a dominating set of $G$. 
Then,  for some path $P_i$, $i \in \{1,\ldots,\ell\}$, if $|P_i|> \xi_{2h} (2 (2h+k_i)+1)$ then $G$ contains a $2h$-{\sc DS}-protrusion ($2h$-{\sc CDS}-protrusion) of size at least $\xi_{2h}$.  Here,  $k_i=|D\cap \Psi(M(P_i))|$. Furthermore, we can find such a  $2h$-{\sc DS}-protrusion ($2h$-{\sc CDS}-protrusion) in polynomial time. 
\end{lemma}
\begin{proof} 
Let $P_i$ be the the path such that $|P_i|> 2 \xi_{2h} (|D\cap \Psi(M(P_i))|$. 
Let 
$P_i:=s_i=a_1^i\cdots a_{\ell}^i=t_i$. For 
every vertex 
\[ w \in \Big(D\cap \Psi(M(P_i))\Big) \cup \kappa(s_is_i^*) \cup \kappa(t_it_i^*) \]
we mark two vertices of the path $P_i$. We mark the first and 
the last vertices on $P_i$, say $a^i_{{\sf wfirst}}$ and $a^{i}_{{\sf wlast}}$, such that $w\in \Psi(a^i_{{\sf wfirst}})$ and $w\in \Psi(a^{i}_{{\sf wlast}})$.  That is, 
$w\in \Psi(a^i_{{\sf wfirst}})$ and $w\in \Psi(a^{i}_{{\sf wlast}})$ and for all $z< {\sf wfirst}$ or $z>{\sf wlast}$ we have that $w\notin \Psi(a^i_{z})$.  
This way we will only mark at most $2(2h+|D\cap \Psi(M(P_i))|)=2(2h+k_i)$ vertices of the path $P_i$. However the path is longer than 
$2 \xi_{2h} (2h+k_i)$  and thus by the pigeonhole principle we have that there exists a subpath of $P_i$, say $P_i(a^{i}_x,a^{i}_y)$, such that no vertex of this subpath is marked and $|P_i(a^{i}_x,a^{i}_y)|>\xi_{2h}$. 
 Let $W=\Psi(M(P_i(a^{i}_x,a^{i}_y)))$. Let $a^*$ and $b^*$ be the neighbors of $a^{i}_x$  and $a^{i}_y$ respectively 
 that are not present on  
$P_i(a^{i}_x,a^{i}_y)$. Clearly, the only vertices in $W$ that have neighbors in $V(G)\setminus W$ belong to 
$\kappa(a^{i}_xa^*)\cup\kappa(a^{i}_yb^*)$.  
Thus the vertices in $W$ that have neighbors in $V(G)\setminus W$  is upper bounded by $2h$. Furthermore, since no vertex on the path $P_i(a^{i}_x,a^{i}_y)$ is marked, we have that all the vertices in $D$ belonging to $W$ also belong to $\kappa(a^{i}_xa^*)\cup\kappa(a^{i}_yb^*)$. Then by Lemma~\ref{lem:smalldominatingset}, we have that 
$\kappa(a^{i}_xa^*)\cup\kappa(a^{i}_yb^*)$ dominates all the vertices in $W$. Furthermore, in $(M,\Psi)$, 
no bag is contained in another and thus $|W|>\xi_{2h}$ (see discussion after Theorem~\ref{thm:structure theorem}). This shows that $W$ is a $2h$-{\sc DS}-protrusion of the desired size.
\end{proof}

The final result of this section is the following decomposition lemma.

\begin{lemma}
\label{lem:slicedeco}
Let $H$ be a fixed graph and ${\cal C}_H$ be the class of graphs that excluds a fixed graph $H$ as a topological minor.  
Then there exist two constants $\delta_1$ and $\delta_2$ (depending on \tDS \, (\tCDS)) 
such that given a yes instance $(G,k)$ of \tDS \, (\tCDS), in polynomial time,  we can either find 
\begin{itemize}
\item a $(\delta_1k,\delta_2 k)$-slice decomposition; or 
\item a $2h$-{\sc DS}-protrusion (or $2h$-{\sc CDS}-protrusion) of size more than $\xi_{2h}$ or;
\item an $h'$-protrusion of size more than $\xi_{h'}$ where $h'$ depends only on $h$.
\end{itemize}
\end{lemma}
\begin{proof}
Let $(G,k)$ be a yes instance of \tDS\  (\tCDS). This implies that the size of the (connected) dominating set $D$ returned by Lemma~\ref{lemma:approximation} (Lemma~\ref{lemma:approximationcds}) is at most $\eta(H)k$. 
Let $M^*$ be the subtree of $M$ formed by  heavy edges. By Lemma~\ref{lemma:boundingleavesandpaths}, we know that 
\begin{itemize}
\item[(a)] $|L(M^*)| \leq \eta(H) k$; 
\item[(b)] $|S_{\geq 3}(M^*)|\leq \eta(H) k-1$; and
\item[(c)] $|\mathscr{P}_2(M^*)|\leq 2 \eta(H) k-1$.  
 \end{itemize}
 Recall that for every path  $P_i \in \mathscr{P}_2(M^*)$, we defined $k_i=|D\cap \Psi(M(P_i))|$. If for any path $P_i \in \mathscr{P}_2(M^*)$ we have that $|P_i|> \xi_{2h} 2 (2h+k_i)$ then by Lemma~\ref{lem:boundingpaths}  $G$ contains a $2h$-{\sc DS}-protrusion of size at least $\xi_{2h}$, and we can find this protrusion in polynomial time. Thus we assume that for all paths $P_i \in \mathscr{P}_2(M^*)$ we have that $|P_i| \leq \xi_{2h} (2 (2h+k_i)+1)$. 

Let $k_i^*$ denote the number of vertices in $D\cap \Psi(M(P_i))$ that are not present in any other 
$D\cap \Psi(M(P_j))$ for $i\neq j$. Furthermore, for all $i\neq j$ we have that 
$$\Big|\Psi(M(P_i))\cap D\cap \Psi(M(P_j) \Big| \leq h.$$
The last assertion is based on the following arguments. The sets $\Psi(M(P_i))$ and  $\Psi(M(P_j))$ can be separated by a separator of size at most $h$ and the vertices of $D$ that appear in  both sets are present in this separator. 
Observe that $k_i\leq 2h+k_i^*$. 
This implies that 
\begin{eqnarray*}
|V(M^*)| & = & |L(M^*)|+|S_{\geq 3}|+|S_{2}| \\
&\leq & \eta(H)k + \eta(H) k-1 + \sum_{P_j\in \mathscr{P}_2(M^*)} (4h+2k_j+2)\xi_{2h} \\
& = & 2 \eta(H)k-1+ (4h+2) |\mathscr{P}_2(M^*)| \xi_{2h} +  \sum_{P_j\in \mathscr{P}_2(M^*)} 2(2h+ k_j^*)\xi_{2h}\\ 
& \leq & 2 \eta(H)k-1+ (8h+2) |\mathscr{P}_2(M^*)|\xi_{2h}+ 2|D| \xi_{2h}\\
& \leq & (2+ (16h+4)\xi_{2h}+2 \xi_{2h})\eta(H)k
\end{eqnarray*}
Let $\Gamma=(2+ (16h+4)\xi_{2h}+2 \xi_{2h})\eta(H)k$. 
This implies that the number of heavy edges is upper bounded by $|{\cal F}|\leq \Gamma-1$. Let $M_1,\ldots,M_\alpha$ be the subtrees of $M$ obtained by deleting all the edges in $M^*$, that is, by deleting all the edges in $\cal F$,  see Fig.~\ref{illus-decompos} for an illustration. Note that \[\alpha\leq \Gamma= (2+ (16h+4)\xi_{2h}+2 \xi_{2h})\eta(H)k.\]  We now argue that either the collection $M_1,\ldots,M_\alpha$ forms a $(\delta_1k,\delta_2 k)$-slice decomposition of $G$ or we have found a $2h$-protrusion or a $2h$-{\sc DS}-protrusion of size more than $\xi_{2h}$ in $G$.

\begin{figure}[t]
\begin{center}
\includegraphics[scale=0.33]{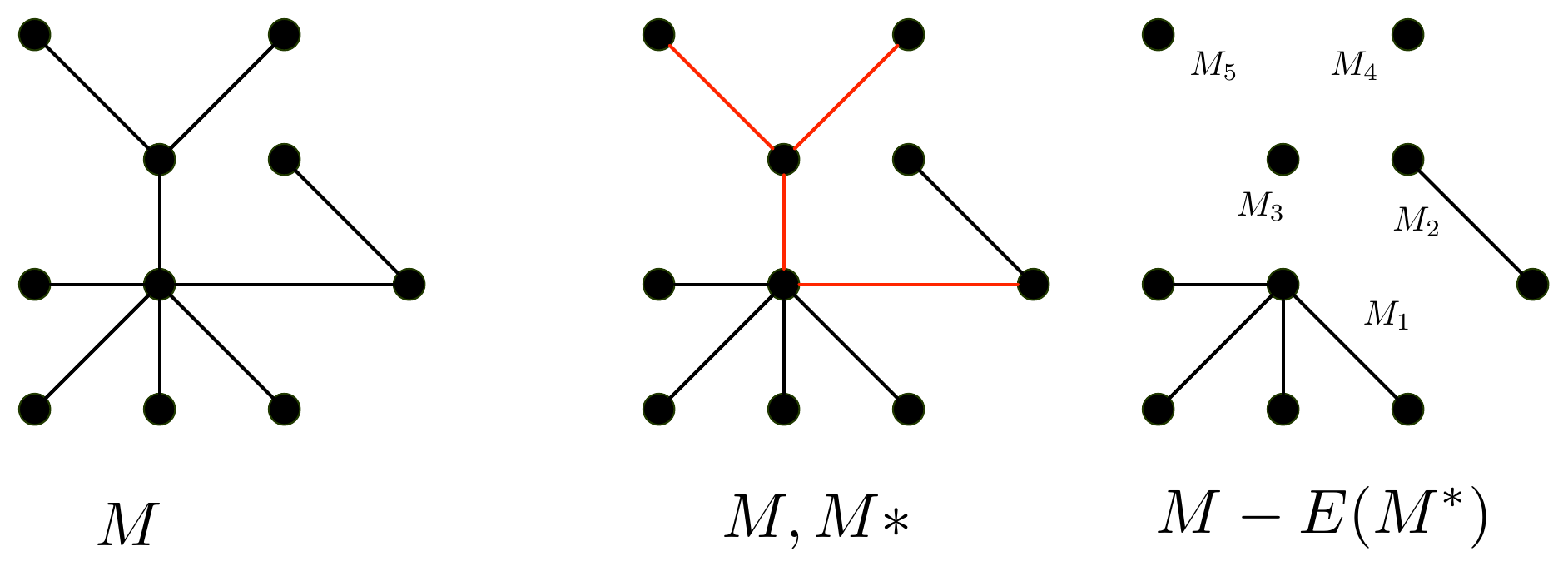}
\end{center}
 \caption{\label{illus-decompos} An illustration of the decomposition. The heavy edges are shown in red.}
\end{figure}

First we show that $$\sum_{i=1}^{\alpha} \Big(\sum_{e\in {\cal E}(M_i)} |\kappa(e)|\Big)=\cO(k).$$ 
Note that by construction, each  $e\in {\cal E}(M_i)$ is a heavy edge. Now observe that each $e$  belongs to at most $2$ distinct edge sets ${\cal E}(M_i)$, we have that $$\sum_{i=1}^{\alpha} \sum_{e\in {\cal E}(M_i)} |\kappa(e)| \leq \left(2 \sum_{e\in E(M^*)={\cal F}} |\kappa(e)| \right)\leq 2h |{\cal F}|\leq 2h\Gamma.$$ 
We set $\delta_2 = 2h(2+ (16h+4)\xi_{2h}+2 \xi_{2h})\eta(H)$, and $\delta_1=\frac{\alpha}{k}$. Since $\alpha = \cO(k)$ we have that $\delta_1$ is a constant; indeed $\alpha\leq \Gamma= (2+ (16h+4)\xi_{2h}+2 \xi_{2h})\eta(H)k$.

Since $M^*$ is connected we have that for every tree $M_i$ there is a unique node in $M_i$ that is incident with  edges in $\cal F$. We denote this special node by $r_i$. We root the tree $M_i$ at $r_i$. Let $w$ be a child of $r_i$ in $M$ and let $M_w$  and $M_{r_i}$ be the subtrees of $M$ obtained after deleting the edge $r_iw$. Since at least one edge incident with  $r_i$ is heavy we have that $\mu(M_{r_i},D)\geq h+1$. However the edge $r_iw$ is not heavy and thus it must be the case that $\mu(M_{w},D)\leq h$. Let $W=\Psi(M_w)$. Then  by Lemma~\ref{lem:smalldominatingset}, we have that  $(D\cap W)\cup \kappa(r_iw)$ is a dominating set of size at most $2h$ for $G[W]$. Furthermore, the only vertices in $W$ that have neighbors in $V(G)\setminus W$ belong to $\kappa(r_iw)$ and thus its size is also upper bounded by $h$. This implies that if $|W|>\xi_{2h}$ then it is a  $2h$-{\sc DS}-protrusion of size at least $\xi_{2h}$. Thus  from now onwards we assume that this is not the case. This implies that for every subtree rooted at $r_i$ and every child $w$ of $r_i$  we have that 
 $|W=\Psi(M_w)|\leq \xi_{2h}$. Next we look at $\tau(r_i)$  and based on its type.  Recall from Theorem~\ref{thm:structure theorem} that they are of the following types. 
 
 \medskip
 
 \noindent 
 {\bf Case 1:  $\tau(r_i)$  has at most $h$ vertices of degree larger than $h$.} 
In the case 
we show that there exists an $h^*$ depending only on $h$ such that either $\tau(M_i)$ has at most $h^* = \xi_{h+\xi_{2h}}+h$ vertices of degree larger than $h^*$,  or $G$ contains an $h'$-protrusion of size more than $\xi_{h'}$. Here, $h'=\xi_{2h}+h$. Suppose some vertex $v$ in $\tau(r_i)$ has degree at most $h$ in  $\tau(r_i)$, but has degree at least $h^*$ in $\tau(M_i)$. Let $N$ be the closed neighbourhood of $v$ in $\tau(r_i)$ and $N'$ be the neighborhood of $v$ in $\tau(M_i)$. 
Each vertex in $N' \setminus N$ must lie in a connected component $C$ of $\tau(M_i) \setminus N$ on at most $\xi_{2h}$ vertices. Towards this, observe that no vertex in $C$ sees any vertex outside $N$ even in the graph $G$.  Thus, if $|C|> \xi_{2h}$ we will get  $2h$-{\sc DS}-protrusion. Let $X$ be $N$ plus the union of all such components. By assumption $|N' \setminus N| \geq  \xi_{h+\xi_{2h}}$ and hence $|X| \geq \xi_{h+\xi_{2h}}$. Finally, the only vertices in $X$ that have neighbors outside of $X$ in $G$ are in $N$, and $|N| \leq h$. The treewidth of $G[X]$ is at most $\xi_{2h}+h$ since removing $N$ from $X$ leaves components of size $\xi_{2h}$. Thus $X$ is an $h'$-protrusion of size more than $\xi_{h'}$. If no such $X$ exists it follows that every vertex of degree at most $h$ in $\tau(r_i)$ has degree at most $h^*$ in $\tau(M_i)$. The vertices of  $\tau(M_i)$ that are not in $\tau(r_i)$ have degree at most $\xi_{2h}+h < h^*$. Thus  $\tau(M_i)$  has at most $h < h^*$ vertices of degree at least $h^*$.

\medskip
 \noindent 
 {\bf Case 2:  $\tau(r_i)$  is $h$-nearly embedded in a surface $\Sigma$ in which $H$ cannot be embedded.} 
In the case  we have that $\tau(r_i)$ excludes some graph $H'$ depending only on $h$ as a minor. The graph $\tau(M_i)$ can be obtained from  $\tau(r_i)$ by joining constant size graphs (of size at most $\xi_{2h}$) to vertex sets that form cliques in $\tau(r_i)$. Thus there exists a graph $H^*$ depending only on $h$ such that $\tau(M_i)$ excludes $H^*$ as a minor. This completes the proof of this lemma. 
\end{proof}




\section{Kernelization Algorithm for \tDS}
\label{sec:domset_kernel}
In this section we use the slice decomposition obtained in the last section 
to obtain linear kernels for  \tDS \ and in the next section outline an algorithm for \tCDS.

Given an instance $(G,k)$ of \tDS \,  we first apply Lemma~\ref{lemma:approximation} and find  a dominating set $D$ of $G$. 
If $|D|>\eta(H)k$ we return that $(G,k)$ is a {\sc no} instance of \tDS. Else, 
we apply Lemma~\ref{lem:slicedeco} and
\begin{itemize}
\item  either find  a $(\delta_1k,\delta_2 k)$-slice decomposition; or 
\item a $2h$-{\sc DS}-protrusion $X$ of $G$ 
of size more than $\xi_{2h}$; or
\item a $h'$-protrusion of size more than $\xi_{h'}$ where $h'$ depends only on $h$.
\end{itemize}
In the second case we apply Lemma~\ref{lem:fiidomset}. Given $X$, by making use of  Lemma~\ref{lem:fiidomset}, we obtain a boundaried graph $X'$ such that $|X'|\leq \xi_{2h}$ and $X\equiv_{\tDS} X'$.  
We also compute the translation constant $c$ between $X$ 
and $X'$.  Now we replace the graph $X$ with $X'$ and obtain a new equivalent instance $(G',k+c)$. See  Definition~\ref{defn:replacement} for the notion of replacement. (Recall that $c$ is a non-positive integer.) In the third case we apply the protrusion replacement lemma of~\cite[Lemma~7]{H.Bodlaender:2009ng} to obtain a new equivalent instance $(G',k')$ for $k' \leq k$ with $|V(G')| < |V(G)|$. We repeat this process until Lemma~\ref{lem:slicedeco}  returns a slice decomposition. 
For simplicity we denote by $(G,k)$ itself the graph on which Lemma~\ref{lem:slicedeco} returns the slice decomposition. Since the number of times this process can be repeated is upper bounded by $n=|V(G)|$, we can obtain a $(\delta_1k,\delta_2 k)$-slice decomposition for $(G,k)$ in polynomial time.

Let $\cal P$ be the pairwise vertex disjoint  subtrees $\{M_1,\ldots,M_\alpha\}$  of $M$ coming 
from the slice decomposition of $G$. Recall that $R_i^+= \Psi(M_i)$.
Let  $Q_i=\bigcup _{e\in {\cal E}(M_i)}\kappa(e)$, $B_i=(D \cap R_i^+)\cup Q_i$ and $b_i=|B_i|$.  
In this section we will treat  $G_i:=G[R_{i}^+]$ as a  graph with boundary $B_i$.  
Observe that by Lemma~\ref{lem:smalldominatingset}, it follows that 
$B_i$ is a dominating set for $G_i$. 

We have two kinds of graphs $G_i$. In one case we have that $G_i$ is $H^*$-minor-free for a graph $H^*$ whose size  depends only on $h$. In the other case we have that the graph $G_i$ has at most $h'$ vertices of degree at least $h'$. To obtain our kernel we will show the following  two lemmas. 

\begin{lemma}
\label{lem:newperspectivequasi}
There exists a constant $\delta$ such that if  $G$  is a graph with boundary $S$ such that $S$ is a dominating set for $G$ and $G$ has at most $h'$ vertices of degree at least $h'$, then in polynomial time, we can obtain a graph $G'$ with boundary $S$ 
such that \[G' \equiv_{\tDS} G \mbox{ and }  |V(G')| \leq \delta |S| .\] 
Furthermore we can also compute the translation constant $c$ of $G$ and $G'$ in polynomial time. 
\end{lemma}

The second lemma is for \Hmf \, graphs.

\begin{lemma}
\label{lem:newperspectiveHminor}
There exists a constant $\delta$ such that  given an \Hmf \, graph $G$  with boundary $S$ 
such that $S$ is a dominating set for $G$ we can obtain, in polynomial time, a graph $G'$ with boundary $S$ such that 
\[G' \equiv_{\tDS} G \mbox{ and }  |V(G')| \leq \delta |S| .\] 
Furthermore we can also compute the translation constant $c$ of $G$ and $G'$ in polynomial time. 
\end{lemma}

Once we have proved Lemmas~\ref{lem:newperspectivequasi} and \ref{lem:newperspectiveHminor}, 
we construct the linear sized kernel for \tDS \, as follows. Given the graph $G$ we 
obtain the slice decomposition and check if any of $G_i$ has size more than $\delta b_i$. (Recall that 
$B_i=(D \cap R_i^+)\cup Q_i$ and $b_i=|B_i|$.) 
 If yes then we either apply  
Lemma~\ref{lem:newperspectivequasi} or Lemma~\ref{lem:newperspectiveHminor} based on the type of $G_i$ 
and obtain a graph $G_i'$ such that  $G_i' \equiv_{\tDS} G_i \mbox{ and }  |V(G_i')| \leq \delta b_i$. We think 
$G=G_i\oplus G^{\star}$, where $G^{\star}=G\setminus (R_i^{+} \setminus B_i)$ as a $b_i$-boundaried graph with 
boundary $B_i$. Then we obtain  a smaller equivalent graph $G'=G^{\star}\oplus G_i' $ and $k'=k+c$.  After this we can 
repeat the whole process once again. This implies that when we cannot apply Lemmas~\ref{lem:newperspectiveHminor} or  
\ref{lem:newperspectivequasi}  on $(G,k)$ we have that each of 
$|V(G_i)|\leq \delta b_i$.  Furthermore notice that by the definition of the slice decomposition we have that $\cup_{i=1}^{\alpha} R_i^{+}=V(G)$. This implies that in this case we have the 
following
\begin{eqnarray*}
\sum_{i=1}^{\alpha} |R_{i}^+| & \leq & \delta \sum_{i=1}^{\alpha} b_i =  \delta (\sum_{i=1}^{\alpha} (|Q_i|+|(D \cap R_i^+)\setminus Q_i|)) \\
& = &
\delta (\sum_{i=1}^{\alpha} |Q_i| +  \sum_{i=1}^{\alpha} |(D \cap R_i^+)\setminus Q_i|) 
\leq \delta \delta_2 k + \delta \eta(H)k = \cO(k). 
\end{eqnarray*}
The last inequality follows from the fact that $\sum_{i=1}^{\alpha} |Q_i|$ is upper bounded by the second component of the slice decomposition and  $\sum_{i=1}^{\alpha} |(D \cap R_i^+)\setminus Q_i|)$ is upper bounded by the size of the approximate dominating set $D$. This brings us to the following theorem. 
\begin{theorem}
\label{thm:lineardomsettopo}
\tDS \, admits a linear kernel on graphs excluding a fixed graph $H$ as a topological minor. 
\end{theorem}

It only remains to prove Lemmas~\ref{lem:newperspectivequasi} and~\ref{lem:newperspectiveHminor} to complete the proof of 
Theorem~\ref{thm:lineardomsettopo}.

\subsection{Irrelevant Vertex Rule and proofs for  Lemmas ~\ref{lem:newperspectivequasi} and \ref{lem:newperspectiveHminor}}
For the proofs of  Lemmas~\ref{lem:newperspectivequasi} and~\ref{lem:newperspectiveHminor} we will introduce a reduction rule that removes irrelevant vertices. 
If the graph $G$ is $K_{h'}$-minor-free then the irrelevant vertex rule  will be used in a recursive fashion. 
In each recursive step it is used in order to reduce the treewidth of torsos and hence also the entire graph. Then the graph 
is split in two pieces and the procedure is applied recursively to the two pieces. In the leaf of the recursion tree when the graph becomes 
smaller but still big enough then we apply Lemma~\ref{lem:fiidomset} on it and obtain an equivalent instance. 

Let $G$ be  a graph given with its  tree-decomposition $(M,\Psi)$ 
as described in Theorem~\ref{thm:structure theorem}, and $\tau(t)$ be one of its torsos. Let $S$ be a dominating set of $G$, and $Z_t=A$, $|A|\leq h$, be the set of apices of $\tau(t)$. The reduction rule essentially ``preserves'' all dominating sets of size at most $|S|$ in $G$, without introducing any new ones. 
To describe the reduction rule we need several definitions. 
The first step in our reduction rule is to classify different subsets $A'$ of $A$ into feasible and infeasible sets. The intuition behind the definition is that a subset $A'$ of $A$ is feasible if there exists a set $D$ in $G$ of size at most $|S|+1$  such that $D$ dominates all but $S$ and   $D \cap A = A'$. However, we cannot test in polynomial time whether such a set $D$ exists. We will therefore say that a subset $A'$ of $A$ is {\em feasible} if the $\eta(H)$-approximation for \tDS \, (Lemma~\ref{lemma:approximation}) 
outputs a set $D$ of size at most $\eta(H)(|S|+2)$ such that $D$  dominates $V(G) \setminus (A \cup S)$ and $D \cap A = A'$. Observe that if such a set $D$ of size at most $|S|+1$ exists then $A'$ is surely feasible in the first sense, while if no such set $D$ of size at most $\eta(H)(|S|+2)$ exists, then $A'$ is surely not feasible (again in the first sense). We will frequently use this in our arguments. Let us remark that there always exists a feasible set $A' \subseteq A$. In particular, $A' = S \cap A$ is feasible since $S$ dominates $G$. For feasible sets $A'$ we will denote by $D(A')$ the set $D$ output by the approximation algorithm.

For every subset $A' \subseteq A$, we select a vertex $v$ of $G$ such that $A'\subseteq N_G[v]$. If such a vertex exists,  we call it a \emph{representative} of $A'$. 
Let us remark that some sets can have no representatives and some distinct subsets of $A$ may have the same representative. 
We define $R$ to be the set of representative vertices for subsets of $A$. The size of $R$ is at most $2^{|A|}$. For $A' \subseteq A$,  the set of {\em dominated vertices} (by $A'$) is $W(A')= N(A') \setminus A$. We say that a  vertex $v \in V(G) \setminus A$ is {\em fully dominated} by $A'$ if $N[v] \setminus A \subseteq W(A')$. A vertex $w \in V(G) \setminus A$ is {\em irrelevant with respect to $A'$} if $w \notin R$, $w \notin S$,  and $w$ is fully dominated by $A'$. 

Now we are ready to state the irrelevant vertex rule. 
\begin{description}
\item[Irrelevant Vertex Rule:] If a vertex $w$ is irrelevant with respect to every feasible $A' \subseteq A$, then delete $w$ from $G$. 
\end{description}

 \begin{lemma}
 \label{lem:domseteqiv1}
Let $S$ be a dominating set in $G$, and $G'$ be the graph obtained by applying the Irrelevant Vertex Rule on $G$, where $w$ was the deleted vertex. Then $G' \equiv_{\tDS} G $.
\end{lemma}
 \begin{proof}
 We view $G$ and $G'$ as graphs with boundary $S$. 
 Let the transposition constant be $0$. To prove that $G' \equiv_{\tDS} G $, we show that given a $|S|$-boundaried graph $G_1$ and a positive 
 integer $\ell$ we have that $(G\oplus G_1,\ell)\in \tDS\, \Leftrightarrow (G'\oplus G_1,\ell)\in \tDS\,$. 
  Let $Z \subset V(G\oplus G_1)$  be a dominating set for $G\oplus G_1$ of size at most $\ell$.  Let $Z_1=V(G)\cap Z$. If $|Z_1|>|S|$ then 
  $(Z\setminus Z_1)\cup S$ is a smaller dominating set for $G\oplus G_1$. Therefore we assume that $|Z_1|\leq |S|$.  Let $A' = Z \cap A$, and observe that $A'$ is feasible because  $Z_1$ dominates all but $S$.  If $w \notin Z$, then $Z'=Z$ is a dominating set of size at most $\ell$ for 
  $G'\oplus G_1$. So assume $w\in Z$. Observe that $w\in Z_1$ and $w\notin S$ and therefore all the neighbors of $w$ lie in $G$. Since $w$ is irrelevant with respect to all feasible subsets of $A$ and $A'$ is feasible, we have that $w$ is irrelevant with respect to $A'$. Hence 
  $N_{G\oplus G_1}(w) \setminus N_{G\oplus G_1}(Z \setminus{w}) \subseteq A$. There is a representative $w' \in R$, $w' \neq w$ (since $w \notin R$), such that $(N_{G\oplus G_1}(w)=N_G(w)) \cap A \subseteq N_G(w') \cap A$. Hence $Z' = (Z \cup \{w'\}) \setminus \{w\}$  is a dominating set of 
  $G'\oplus G_1$ of size at most $\ell$. 
  
Now, let $Z' \subseteq V(G'\oplus G_1)$  be a dominating set of size at most $\ell$ for $G'\oplus G_1$.  Let  $Z_1'=V(G')\cap Z'$. As in the forward direction we can assume that $|Z_1'|\leq |S|$. We show that $Z'$ also dominates $w$ in $G\oplus G_1$. Specifically $Z_1' \cup \{w\}$ is a set dominating all but $S$ in $G$ of size at most $|S|+1$ so $Z_1' \cap A$ is feasible. Since $\{w\}$ is irrelevant with respect to $Z_1' \cap A$, we have $w \in N_G(Z_1' \cap A)$ and thus $Z'$ is a dominating set for $ G\oplus G_1$ of size at most $\ell$. This concludes the proof. 
 \end{proof}

For a graph $G$ and its dominating set $S$, we apply the Irrelevant Vertex Rule exhaustively on all torsos of $G$, obtaining an induced subgraph $G'$ of $G$. By Lemma~\ref{lem:domseteqiv1} and transitivity of $\equiv_{\tDS}$ we have that $G'\equiv_{\tDS}G$. We now prove that a graph $G$ which can not be reduced by the irrelevant vertex rule has a property that each of its torso has a small $2$-dominating set. 


\begin{lemma}
\label{lem:dstwbound} Let $G$ be a graph which is irreducible by the Irrelevant Vertex Rule and $S$ be a dominating set of $G$. 
For every torso $\tau(t)$ of the tree-decomposition $(M,\Psi)$ of $G$, we have that $\tau(t)\setminus Z_t$ 
has a $2$-dominating set of size $\cO(|S|)$. Furthermore if $G$ is a \Hmf \, graph then $\tw(G)=\cO(\sqrt{|S|})$. 
\end{lemma}
\begin{proof}
Let $\tau(t)^*= \tau(t) \setminus A$, where $A$ are the apices of $\tau(t)$. We will obtain a $2$-dominating set of size $\cO(|S|)$ in $\tau(t)^*$. Towards this end, consider the following set, $$Q=\bigcup_{A' \subseteq A, A' \text{is feasible} }D(A')\cup R\cup (S \setminus A).$$ 
The number of   representative vertices $R$ and the number of feasible subsets $A'$ 
 is at most $2^{|A|}\leq 2^h$, 
  where $h$ is a constant depending only on $H$. The size of $D(A')$ is at most $\eta(H)(|S|+2)$ for every $A'$. Thus  $|Q|\leq 2^h (\eta(H)(|S|+2)) + 2^h + |S|=\cO(|S|)$. We prove that $Q$ is a $2$-dominating set of $V(G) \setminus A$. Let $w \in V(G) \setminus A$. If $w \in R$ or $w \in S$, then $Q$ dominates $w$. So suppose $w \notin R \cup S$. Then, since $w$ is not irrelevant, we have that  there is a feasible subset $A'$ of $A$ such that $w$ is relevant with respect to $A'$. Hence $w$ is not fully dominated by $A'$ and so $w$ has a neighbour $w' \in V(G) \setminus N[A']$. But $w'$ is dominated by $D(A') \subseteq Q$, and thus $w$ is $2$-dominated by $Q$ in $G \setminus A$. Hence $G \setminus A$ has a $2$-dominating set of size $\cO(|S|)$.

The graph $\tau(t)^*$ can be obtained from $G \setminus A$ by contracting all edges in $E(G \setminus A) \setminus E(\tau(t)^*)$ and adding all edges in $E(\tau(t)^*) \setminus E(G \setminus A)$. Since contracting and adding edges does  not increase the size of a minimum $2$-dominating set of a graph, $\tau(t)^*$ has a $2$-dominating set of size $\cO(|S|)$. This completes the proof for the first part.

Now assume that $G$ is a \Hmf \, graph. It is well known  that the treewidth of a \Hmf \, graph is at most the maximum treewidth of its torsos, see e.g.\cite{DemaineFHT05sub}. Thus to show that $\tw(G)=\cO(\sqrt{|S|})$ it is sufficient to show that its torsos have small treewidth.  To conclude, $\tau(t)^*$ excludes an apex graph as a minor  (see, e.g.~\cite[Theorem $13$]{Grohe03}) and it has a $2$-dominating set of size $\cO(|S|)$. By the bidimensionality of $2$-dominating set, we have that $\tw(\tau(t)^*)=\cO(\sqrt{|S|})$~\cite{DemaineFHT05sub,FominGT09con}. Now we add all the apices of $A$ to all the bags of the tree-decomposition  of $\tau(t)^*$ to obtain a tree-decomposition for $\tau(t)$ of width   $ \cO(\sqrt{|S|})+h =\cO(\sqrt{|S|})$.  
\end{proof}

Let us also remark that Irrelevant Vertex Rule is based on the performance of a polynomial time approximation algorithm. Thus by 
Lemmas~\ref{lemma:approximation}, \ref{lem:domseteqiv1} and ~\ref{lem:dstwbound}, and the fact that the treewidth of a graph is at most the maximum treewidth of its torsos, see e.g.\cite{DemaineFHT05sub}, we obtain the following lemma.
\begin{lemma}
\label{lem:sumreductiondomset}
There is a polynomial time algorithm that for  a given graph $G$ and a dominating set $S$ of $G$, outputs graph $G'$ such that $G' \equiv_{\tDS} G $ and for every torso $\tau(t)$ of the tree-decomposition $(M,\Psi)$ of $G$, we have that $\tau(t)\setminus Z_t$ has a $2$-dominating set of size $\cO(|S|)$. Furthermore if $G$ is a \Hmf \, graph then $\tw(G)=\cO(\sqrt{|S|})$. 
\end{lemma}

Before we proceed further, we show the power of Lemma~\ref{lem:sumreductiondomset} by deriving a simple subexponential time algorithm for \tDS \, on \Hmf \, graph. This is one of the cornerstone results in~\cite{DemaineFHT05sub} and is based on a non-trivial two-layer dynamic programming over clique-sum decomposition tree of a \Hmf \, graphs. 
Lemma~\ref{lem:sumreductiondomset} can be used to obtain much simpler algorithm. Given a graph $G$ and a positive integer $k$ we first apply a factor $2$-approximation algorithm given in~\cite{DemaineHaj05,FominLRS10} for \tDS \, on $G$ and obtain a set $S$. If the size of $S$ is more than $2k$ then we return that $G$ does not have a dominating set of size at most $k$. Otherwise,  we apply Lemma~\ref{lem:sumreductiondomset} and obtain an equivalent graph $G'$ such that $\tw(G') = \cO(\sqrt{k})$. Now applying a constant factor approximation algorithm developed in~\cite{DemaineFHT05sub} for computing the treewidth on $G'$ we get a tree-decomposition of width  $\cO(\sqrt{k})$. It is well known that checking whether a graph with treewidth $t$ has a dominating set of size at most $k$ can be done in time $\cO(3^{t} n^{\cO(1)})$ ~\cite{RooijBR09}. This together with the above bound on the treewidth, gives us an alternative proof of the following theorem.
\begin{theorem}[\cite{DemaineHaj05}]\label{THM:Demaine}
Given an $n$-vertex  graph $G$ excluding a fixed graph $H$ as a minor, one can check whether $G$ has a dominating set of size at most $k$ in time $2^{\cO(\sqrt{k})}n^{\cO(1)}$. 
\end{theorem}

Having Lemma \ref{lem:sumreductiondomset} proving Lemma~\ref{lem:newperspectivequasi} becomes simple.

\begin{proof}[Proof of Lemma~\ref{lem:newperspectivequasi}]
We apply Lemma~\ref{lem:sumreductiondomset} to $G$ with a decomposition that has a single bag containing the entire graph and the apices $A$ of the bag being the vertices of degree at least $h'$. By Lemma~\ref{lem:sumreductiondomset}, $G \setminus A$ has a $2$-dominating set of size $\delta_3|S|$. Since all vertices of $G \setminus A$ have degree at most $h'$ it follows that $|V(G)| \leq h' +  \delta_3 |S| + \delta_3h|S|+ \delta_3h^2|S| \leq \delta|S|$. 
\end{proof}

We need the following well known lemma, see e.g.\ \cite{Bodlaender98}, on separators in graphs of bounded treewidth for the proof of Lemma~\ref{lem:newperspectiveHminor}. 
\begin{lemma}
\label{lemma:balsep1}
Let $G$ be a graph given with a tree-decomposition of width at most $t$ and 
$w:V(G)\rightarrow \{0 ,1\}$ 
  be a weight function. 
Then  in polynomial time we can find a bag $X$ of the given tree-decomposition 
such that for every connected component $G[C]$ of $G\setminus X$, $w(C)\leq w(V(G))/2$. Furthermore, 
the connected components $C_1,\ldots,C_\ell$ of $G \setminus X$ can be grouped into two sets 
$V_1$ and $V_2$ such that 
$\frac{w(V(G))-w(X)}{3}\leq w(V_i)\leq \frac{2(w(V(G))-w(X))}{3}$, for $i\in \{1,2\}$. 
\end{lemma}

\begin{proof}[Proof of Lemma~\ref{lem:newperspectiveHminor}]
By $(G,S)$ we denote the graph with boundary $S$. 
By Lemma~\ref{lem:sumreductiondomset},  we may assume that $\tw(G) = \cO(\sqrt{|S|})$. We prove the lemma using induction on $|S|$. 
If $|S|=\cO(1)$ we are done, as in this case we know that $G$ is a $|S|$-{\sc DS} protrusion. Thus, if $|V(G)|>\xi_{|S|}$ then 
 we can apply Lemma~\ref{lem:fiidomset} and in polynomial time obtain a graph $G^*$ such that $G^*\equiv_{\tDS} G$ and 
 $|V(G^*)|\leq \xi_{|S|}$. In the same time we can compute the translation constant depending on $G$ and $G^*$ 
 and return it. Thus, we return $G^*$  and the translation constant $c$.

Otherwise, using a constant factor approximation of treewidth on \Hmf \, graphs~\cite{FeigeHajLee08}, we compute a tree-decomposition of $G$ of width $d\sqrt{|S|}$, for some constant $d$. Now, by applying Lemma~\ref{lemma:balsep1} on this decomposition, we find a partitioning of $V(G)$ into $V_1$, $V_2$ and $X$ such that there are no edges from $V_1$ to $V_2$, $|X| \leq d\sqrt{|S|}+1$, and $|V_i \cap S| \leq 2|S|/3$ for $i\in \{1,2\}$. Let 
$S' = S \cup X$. Observe that $S'$ is also a dominating set. 

Let $S_1=S' \cap (V_1 \cup X)$ and $S_2=S' \cap (V_2 \cup X)$.  Let $G_1=G[V_1 \cup X]$ and $G_2=G[V_2 \cup X]$. 
We now apply the algorithm recursively on 
$(G_1, S_1)$ and $(G_2, S_2)$ and obtain graphs 
$G_1'$, $G_2'$  such that for $i \in \{1,2\}$,  $G_{i}\equiv_{\tDS} G_i$. Let $c_1$ and $c_2$ be the translation 
constants returned by the algorithm.  Since $X \subseteq S'$, we have that $S_i$ is   a dominating set of $G_i$ and hence we actually can run the algorithm recursively on the two subcases. The algorithm returns $G_1'$ and $G_2'$ and translation constants $c_1$ and $c_2$. Let 
$G'=G_1'\oplus_\delta G_2'$ and $S'=S_1\cup S_2$.  We will show that $G'\equiv_{\tDS} G$. Let $G_3$ be a graph with boundary $S'$ and $k$ be a positive integer.  Then
\begin{eqnarray*}
((G_1\oplus_\delta G_2)\oplus G_3, k) & \in & \tDS \\
\iff  ((G_1\oplus_\delta G_3)\oplus G_2, k) & \in & \tDS \\
\iff  ((G_1\oplus_\delta G_3)\oplus G_2', k+c_2) & \in & \tDS \\
\iff  ((G_2'\oplus_\delta G_3)\oplus G_1, k+c_2) & \in & \tDS \\
\iff  ((G_2'\oplus_\delta G_3)\oplus G_1', k+c_2+c_1) & \in & \tDS \\
\iff  ((G_2'\oplus_\delta G_1')\oplus G_3, k+c_2+c_1) & \in & \tDS. \\
\end{eqnarray*}
This proves that $G'\equiv_{\tDS} G$.  Now we will show that  $|V(G')|\leq \cO(|S|)$. 

Let $\mu(|S|)$ be the largest possible
size of the set $|V(G')|$ output by the algorithm when run on a graph $G$ with a
dominating set $S$. We
upper bound $|V(G')|$ by the following recursive formula.
 \begin{eqnarray*}
 |V(G')| & \leq & \max_{1/3 \leq \alpha \leq 2/3 }\left\{ \mu\left(\alpha |S|  + d \sqrt{|S|}\right) 
	    +  \mu\left( (1-\alpha)|S|  \right) 
	   +  d \sqrt {|S|}  \right\}.
\end{eqnarray*}
Using simple induction one can show that the above solves to $\cO(|S|)$. See for an example~\cite[Lemma~$2$]{FominLRS10}. Hence we conclude that $|V(G')|=\cO(|S|)=\cO(k)$. This completes the proof of the lemma. 
\end{proof}

The algorithm of Demaine et al. \cite{DemaineHaj05} computing a dominating set of size $k$ in an  $n$-vertex  \Hmf \, graph uses exponential (in $k$) space
$2^{\cO(\sqrt{k})} n^{\cO(1)}$. 
Theorem~\ref{thm:lineardomsettopo} implies almost directly the following refinement of Theorem~\ref{THM:Demaine}.
  \begin{theorem} 
 Given an $n$-vertex  graph $G$ excluding a fixed graph $H$ as a minor,  one  can check whether $G$ has a dominating set of size at most $k$ in time 
$2^{\cO(\sqrt{k})}+ n^{\cO(1)}$ and space $(nk)^{\cO(1)}$. 
\end{theorem}
\begin{proof} 
%

Our algorithm first applies Theorem~\ref{thm:lineardomsettopo} to obtain a graph with $O(k)$ vertices.
Now we are assuming that the number of vertices in $G$ is $n=\cO(k)$. 
We solve a slightly more general version of domination, where we are   given a subset $S$ and the requirement is to find a set $D$ of size at most   $k$ such that for every $v\in V(G) \setminus S$, $N[v] \cap D \neq \emptyset$. When $S=\emptyset$, the set $D$ is a dominating set of size $k$.
By the separator theorem of Alon et al.  \cite{AlonST90}
 for \Hmf \, graphs, one can find in polynomial time a partition of $V(G)$ into $V_1$, $V_2$ and $X$ such that $|X| \leq \cO(\sqrt{n})$, there are no edges from $V_1$ to $V_2$ and $|V_i| \leq 2n/3$ for $i \in \{1,2\}$. The algorithm finds such a partition and guesses how $D$ interacts with $X$. 

In particular, first the algorithm correctly guesses $D' = D \cap X$ (by looping over all subsets of $X$). For each guess, it puts $N(D')$ into $S$ and removes $D'$ and $S \cap X$ from $G$ (these vertices are already dominated and will not be used in the future to dominate even more vertices). For every remaining vertex $v$ in $X$, the algorithm guesses whether it will be dominated by a vertex in $V_1$, in which case the algorithm deletes all edges from $v$ to vertices in $V_2$, or by a vertex in $V_2$, in which case the algorithm deletes all edges from $v$ to vertices in $V_1$. Let $V_i'$ be $V_i$ plus all the vertices in $X \setminus S$ that we guessed were dominated from $V_i$. At this point $V_1'$ and $V_2'$ are distinct components of the instance and can be solved independently. The running time is governed by the following recurrence.
$$T(n) =n^{\cO(1)} \cdot  2^{\cO(\sqrt{n})} \cdot 2 \cdot T(2n/3) =  2^{\cO(\sqrt{n})}$$
The space used is clearly polynomial. This concludes the proof.
\end{proof}

\section{Kernelization algorithm for \tCDS}\label{sec:CDSkernel}
The kernelization algorithm for \tCDS \, is   similar to \tDS---we also  use slice decomposition to obtain a 
linear kernel.  However, the irrelevant vertex rule is a bit different. 
The  kernelization algorithm for \tCDS \, follows from the  results 
analogous to  Lemmas~\ref{lem:newperspectiveHminor} and \ref{lem:newperspectivequasi} for \tDS.  For completeness we spell out all the steps.

In particular given an instance $(G,k)$ of \tCDS \,  we first apply Lemma~\ref{lemma:approximation} and find  a dominating set $D$ of $G$. 
If $|D|>\eta(H)k$ we return that $(G,k)$ is a {\sc no} instance of \tCDS. Else, 
we apply Lemma~\ref{lem:slicedeco} and
\begin{itemize}
\item  either find  $(\delta_1k,\delta_2 k)$-slice decomposition; or 
\item a $2h$-{\sc CDS}-protrusion of size more than $\xi_{2h}$; or
\item a $h'$-protrusion of size more than $\xi_{h'}$ where $h'$ depends only on $h$.
\end{itemize}
In the second case we apply Lemma~\ref{lem:fiidomset}. For a given $X$, we apply Lemma~\ref{lem:fiidomset} and construct a boundaried graph $X'$ such that $|X'|\leq \xi_{2h}$ and $X\equiv_{\tCDS} X'$. We also compute the translation constant $c$ between $X$ 
and $X'$.  Now we replace the graph $X$ with $X'$ and obtain a new equivalent instance $(G',k+c)$, here we remind that $c$ is a  non-positive integer. In the third case we apply the protrusion replacement lemma of~\cite[Lemma~7]{H.Bodlaender:2009ng} to obtain a new equivalent instance $(G',k')$ for $k' \leq k$ with $|V(G')| < |V(G)|$. We repeat this process until Lemma~\ref{lem:slicedeco}  returns a slice decomposition. 
For simplicity we denote by $(G,k)$ itself the graph on which Lemma~\ref{lem:slicedeco} returns the slice decomposition. The number of times this process can be repeated  does not exceed  $n=|V(G)|$ and a  $(\delta_1k,\delta_2 k)$-slice decomposition for $(G,k)$ is constructed  in polynomial time.   

 The pairwise disjoint connected subtrees $\{M_1,\ldots,M_\alpha\}$  of $M$ coming 
from the slice decomposition of $G$ is denoted by  $\cal P$ and we put  $R_i^+= \Psi(M_i)$.
We define   $Q_i=\bigcup _{e\in {\cal E}(M_i)}\kappa(e)$, $B_i=(D \cap R_i^+)\cup Q_i$ and $b_i=|B_i|$.  
As in the previous section,  we   treat  $G_i:=G[R_{i}^+]$ as a  graph with boundary $B_i$.  
Then by Lemma~\ref{lem:smalldominatingset},  
$B_i$ is a dominating set for $G_i$. 

For  two kinds of graphs $G_i$, we use different reductions. 
 In the first case we have that the graph $G_i$ has at most $h'$ vertices of degree at least $h'$.
\begin{lemma}
\label{lem:newperspectivequasicds}
There exists a constant $\delta$ such that if  $G$  is a graph with boundary $S$  such that $S$ is a dominating set for $G$ and $G$ has at most $h'$ vertices of degree at least $h'$, then in polynomial time, we can obtain a graph $G'$ with boundary $S$ 
such that \[G' \equiv_{\tCDS} G \mbox{ and }  |V(G')| \leq \delta |S| .\] 
Furthermore we can also compute the translation constant $c$ of $G$ and $G'$ in polynomial time. 
\end{lemma}
  In the other  case we have that $G_i$ is $H^*$-minor-free for a graph $H^*$ whose size only depends on $h$.  \begin{lemma}
\label{lem:newperspectiveHminorcds}
There exists a constant $\delta$ such that  given an \Hmf \, graph $G$  with boundary $S$ 
such that $S$ is a dominating set for $G$, in polynomial time, we can obtain a graph $G'$ with boundary $S$ such that 
\[G' \equiv_{\tCDS} G \mbox{ and }  |V(G')| \leq \delta |S| .\] 
Furthermore we can also compute the translation constant $c$ of $G$ and $G'$ in polynomial time. 
\end{lemma}

In order to obtain the linear sized kernel for \tCDS \, the proof of  
 Lemmas~\ref{lem:newperspectivequasicds} and \ref{lem:newperspectiveHminorcds} suffices. 
 Indeed, for  graph $G$ we 
obtain the slice decomposition and check if any  $G_i$ has size more than $\delta b_i$. If yes then we either apply  
Lemma~\ref{lem:newperspectivequasicds} or Lemma~\ref{lem:newperspectiveHminorcds} based on the type of $G_i$ 
and obtain a graph $G_i'$ such that  $G_i' \equiv_{\tCDS} G_i \mbox{ and }  |V(G_i')| \leq \delta b_i$. We view 
$G=G_i\oplus G^{\star}$, where $G^{\star}=G\setminus (R_i^{+} \setminus B_i)$ as a $b_i$-boundaried graph with 
boundary $B_i$. Then we obtain  a smaller equivalent graph $G'=G^{\star}\oplus G_i' $ and $k'=k+c$.  After this we can 
repeat the whole process once again. This implies that when we can not apply Lemmas~\ref{lem:newperspectiveHminorcds} or  
\ref{lem:newperspectivequasicds}  on $(G,k)$ we have that each of 
$|V(G_i)|\leq \delta b_i$.  Furthermore notice that $\cup_{i=1}^{\alpha} R_i^{+}=V(G)$. This implies that  \begin{eqnarray*}
\sum_{i=1}^{\alpha} |R_{i}^+| & \leq & \delta \sum_{i=1}^{\alpha} b_i =  \delta (\sum_{i=1}^{\alpha} (|Q_i|+|(D \cap R_i^+)\setminus Q_i|)) \\
& = &
\delta (\sum_{i=1}^{\alpha} |Q_i| +  \sum_{i=1}^{\alpha} |(D \cap R_i^+)\setminus Q_i|) 
\leq \delta \delta_2 k + \delta \eta(H)k = \cO(k). 
\end{eqnarray*}

Thus (subject to the proof of two lemmas) we have the following theorem. 
\begin{theorem}
\label{thm:lineardomsettopocds}
\tCDS \, admits a linear kernel on graphs excluding a fixed graph $H$ as a topological minor. 
\end{theorem}
%

\subsection{Irrelevant Vertex Rule and proofs for  Lemmas~\ref{lem:newperspectivequasicds} and \ref{lem:newperspectiveHminorcds} }

As with \tDS, we will reduce  the treewidth of a torso not only in the beginning of the procedure but also when we apply it recursively. 
Let $G$ be an \Hmf \,  graph, $S$ be a dominating set of $G$ (not necessarily connected), $\tau(t)$ be one of its torsos, and $A$, $|A|\leq h$,  be the set of apices of $\tau(t)$, where $h$ is some constant depending only on $H$. We will define a reduction rule that essentially ``preserves" all dominating sets of size at most $3|S|+3$ with ``good enough'' connectivity properties, without introducing new such sets. Just as for \tDS{} we will say that a subset $A'$ of $A$ is feasible if the factor $\eta(H)$-approximation for \tDS{} (Lemma~\ref{lemma:approximation})  concludes that there exists a set $D$ of size 
at most $\eta(H)(3|S|+3)$ which   dominates $V(G) \setminus (A \cup S)$ and $D \cap A = A'$. 
If such a set exists and $A'$ is feasible we denote this set by $D(A')$.


Recall, that for \tDS \, we had the notion of a representative element for every subset $A'\subseteq A$. The representative vertex was crucially used in establishing  Lemma~\ref{lem:domseteqiv1}, where we used it to simulate all the domination properties of the deleted vertex ``$w$''.  We need a similar notion of representatives for \tCDS, however here the representatives will be vertex subsets rather than single vertices. With vertex subsets we will be able to simulate not only   domination properties, but also the connectivity properties of an  irrelevant vertex. More precisely, for every subset $A' \subseteq A$, we compute a minimum size vertex set $T \subseteq V(G)\setminus A$ such that $G[T]$ is connected and $A' \subseteq N[T]$. If the size of such a minimum set is at most $4h$, then we say that $T=T(A')$ is a {\em representative} of $A'$, and add all the vertices in $T$ to the set $R$. Note that $|R| \leq 4h \cdot 2^h$.
For each   $A'$ we can test whether a representative exists in time $2^{|A'|}n^{\cO(1)}=2^hn^{\cO(1)}$ by making a modification of the algorithm for the  Steiner tree problem from~\cite{BjorklundHKK07}. Alternatively we can test it in time $n^{4h+\cO(1)}$ by brute force. Let $S_{4h}$ denote the set of vertices in 
$N_{G\setminus A}^{4h}[S]=N_{G\setminus A}^{4h}[S\setminus A]$. Here $N_{G\setminus A}^{4h}[w]$ is the set of vertices at distance at most $4h$ from $w$ in the graph $G\setminus A$ (not in $G$). 
The set of vertices \emph{covered} by $A'$ is $W(A') = N[A']\setminus (A \cup S \cup  S_{4h})$. Note that a vertex in $N_{G\setminus A}^{4h}[S]$ is never covered by a set $A'$.  Let ${\sf CutVert}$ denote the set of vertices  $w$ in $G$ such that $G-\{w\}$ has more connected components than $G$. Observe that if $G$ will be connected  then ${\sf CutVert}$ is essentially the set of cut vertices. However, for a disconnected graph it is the union of cut vertices for each connected component. 

The definition of an irrelevant vertex with respect to $A$ is   different than for \tDS{}. A vertex
 \[ w\notin (S\cup S_{4h} \cup R \cup {\sf CutVert})\]
  is called {\em irrelevant with respect to $A'$}, if $N_{G\setminus A}^{4h}[w] \subseteq W(A')$.  The irrelevant vertex rule for \tCDS \, is exactly the same as in Section~\ref{sec:domset_kernel} for \tDS \, but the correctness proof and analysis is more complicated. Recall that a subset $A'$ of $A$ is feasible if the factor $\eta(H)$-approximation for \tDS{} (Lemma~\ref{lemma:approximation})  concludes that there exists a set $D$ of size 
at most $\eta(H)(3|S|+3)$ which dominates all but $S$, such that $S \cap A = A'$.

\begin{description}
\item[Irrelevant Vertex Rule:] If a vertex $w$ is irrelevant with respect to every feasible $A' \subseteq A$ then delete $w$ from $G$. 
\end{description}

 \begin{lemma}
 \label{lem:condomseteqiv1}
Let $S$ be a dominating set in $G$, and $G'$ be the graph obtained by applying the Irrelevant Vertex Rule on $G$, where $w$ was the deleted vertex.  Then $G' \equiv_{\tCDS} G$. 
\end{lemma}
 \begin{proof}
 We view $G$ and $G'$ as graphs with boundary $S$. Let the transposition constant be $0$. To show that $G' \equiv_{\tCDS} G$, we show that given any boundaried graph $G_1$ and a positive 
 integer $\ell$ we have that $(G\oplus G_1,\ell)\in \tCDS\, \Leftrightarrow (G'\oplus G_1,\ell)\in \tCDS\,$.  Let $Z \subset V(G\oplus G_1)$  be a connected dominating set for $G\oplus G_1$ of size at most $\ell$.  Observe that since $S$ is a dominating set of $G$, we have that there exists a connected dominating set $S\subseteq S^*$  such that $|S^*|\leq 3|S|$ (Proposition~\ref{lem:bb}). Let $Z_1=V(G)\cap Z$. If $|Z_1|>3|S|$ then 
  $(Z\setminus Z_1)\cup S^*$ is a smaller connected dominating set for $G\oplus G_1$. Thus, we assume that 
  $|Z_1|\leq 3|S|$.  Let $A' = Z_1 \cap A$, and observe that $A'$ is feasible since $Z_1$ dominates all but $S$ and has size at  most $3|S|$. If $w \notin Z$, then $Z'=Z$ is a connected dominating set of size $\ell$ for $G'\oplus G_1$.  So assume $w\in Z$. Since $w$ is irrelevant with respect to $A'$ we have that $N_{G\setminus A}^{4h}[w] \subseteq W(A')$.


Let $Q$ be the connected component of $G\oplus G_1$ that contains $w$.  Since, $w$ is not a cut vertex of $G$, we have the following easy observation. 

\begin{observation}
\label{claim:wisokay}
$Q\setminus \{w\}$ is connected. 
\end{observation}

Let $Z_Q=Z \cap Q$ be the connected dominating set of $Q$, $|Z_Q|=p$. 
We will show that $Q\setminus \{w\}$ has a connected dominating set of size at most $p$ and that will show that 
$(G'\oplus G_1,\ell)\in \tCDS$.   Observe that since $w\in W(A')$ and the only vertices that are common between $G$ and $G_1$ belong to $S$, we have that 
$N_{(G\oplus G_1)\setminus A}^{4h}[w]=N_{G\setminus A}^{4h}[w]\subseteq V(G)\setminus S_{3h}$. 

Let $X$ be the vertex set of the connected component of $G\oplus G_1[Z_Q \cap N_{G\setminus A}^{4h}[w]]$ that contains $w$. 
If $|X| < 4h$ then there is a subset $X' = T(N(X) \cap A)$ such that $X' \subset R$, $|X'| \leq |X|$, $G[X']$ is connected and 
$N_G(X') \cap A \supseteq N_G(X) \cap A$.  Furthermore, since $|X| < 4h$ we have that every connected component of $G\oplus G_1[Z_Q\setminus X]$ contains a vertex of $A'$. This implies that  $Z_Q' = (Z_Q \setminus X) \cup X'$ is connected. Since $X \subseteq W(A')$,  and $|X| < 4h$ we have that $N_{G\oplus G_1}(X)=N_{G}(X)$. This implies that 
$N_G(X) \subseteq N_{G}(X'\cup A') \subseteq N_{G\oplus G_1}(X'\cup A') $ and thus $Z_Q'$ is a connected dominating set of size at most $p$ of $Q$ that avoids $w$ and thus by Observation~\ref{claim:wisokay}, it  is also a connected dominating set of $Q\setminus \{w\}$. This implies that in this case $(G'\oplus G_1,\ell)\in \tCDS$.

Now suppose that $|X| \geq 4h$. Let $A^* = N_G(X) \cap A$. The vertex set $A^*$ is a dominating set of size at most $h$ in the connected graph $G[A^* \cup X]$ and so $G[A^* \cup X]$ has a connected dominating set $X^*$ that contains $A^*$ of size at most $3h$. Let $P$ be the connected component of $G[X^*] \setminus A$ that contains $w$. Notice that $|P| \leq 2h$ and so there is a connected set $P' \subseteq R$ such that $|P'|\leq |P|$ and $N(P) \cap A \subseteq N(P') \cap A$. Finally, let $Y$ be the set of vertices in $X$ that are at distance exactly $4h$ from $w$ in $G \setminus A$. Note that $|X \setminus Y| \geq 4h-1$ (as every path from $w$ to  a vertex in $Y$ has length at least $4h-1$) and that $N_G[Y] \cap A \subseteq A^*$. Set 
$X' = (X^* \setminus P) \cup P'$, and $Z_{Q}' = (Z_{Q} \setminus (X \setminus Y)) \cup X'$. We have that $|X'| \leq |X^*| \leq 3h$ while $|X \setminus Y| \geq 4h-1 \geq 3h$. Hence $|Z_Q'| \leq |Z_Q|$.  Note that $G[X']$ is connected. Furthermore by our choice of 
$(X\setminus Y)$ we have that every connected component of $G\oplus G_1[Z_Q\setminus X]$ contains a vertex of 
$Y$ and hence a vertex of $A^*$. However, $A^*\subseteq X'$ and $G[X']$ (or $G\oplus {G_1}[X']$)  is connected and thus  
$G\oplus G_1[Z_Q']$ is connected.  Observe that $N_{G\oplus G_1}(X\setminus Y)=N_{G}(X\setminus Y)$. This implies that 
$N_G(X\setminus Y) \subseteq N_{G}(X'\cup A^*) \subseteq N_{G\oplus G_1}(X'\cup A^*) $ and thus $Z_Q'$ is a connected dominating set of size at most $p$ of $Q$ that avoids $w$ and thus by Observation~\ref{claim:wisokay} is also a connected dominating set of $Q\setminus \{w\}$. This implies that in this case $(G'\oplus G_1,\ell)\in \tCDS$.




Now we prove the reverse direction. Let  $Z' \subset V(G' \oplus G_1)$  be a connected dominating set for $G'\oplus G_1$ of size at most $\ell$.  By Observation~\ref{claim:wisokay} we know that $Q\setminus \{w\}$ and $Q$ are connected and thus $Z'$ is also a connected dominating set of size at most $\ell$ for $G\oplus G_1$. This concludes the proof.  
 \end{proof}

 Next we prove an auxiliary lemma that upper bounds the number of cut vertices in terms of the dominating set of the graph. 
 
 \paragraph{Cuts and Blocks.} 
 A maximal connected subgraph without a cut vertex is called a \textbf{block}. Every block of a graph $G$ is either a maximal $2$-connected subgraph, or a bridge or an isolated vertex. By maximality, different blocks of $G$ overlap in at most one vertex, which is then a cut vertex of $G$. Therefore, every edge of $G$ lies in a unique block and $G$ is the union of its blocks.

\begin{definition} Let $A$ denote the set of cut vertices of $G$ and $B$ the set of its blocks. The bipartite graph on $A\cup B$ where $a\in A$ and $b\in B$ are adjacent when $a\in b$ is called the block graph of $G$.
\end{definition}

\begin{proposition}[\cite{diestelbook}] 
\label{prop:blockgraphtree}
The block graph of a connected graph is a tree.
\end{proposition}

 \begin{lemma}
 \label{lem:cutvertexcount}
 Let $G$ be a graph and $S$ be a dominating set of $G$, then the number of cut vertices in $G$ is upper bounded by $|S|$. That is, $|{\sf CutVert}|\leq |S|$. 
 \end{lemma} 
\begin{proof}
Let $A={\sf CutVert}$ denote the set of cut vertices of $G$ and $B$ the set of its blocks. Consider the block graph $\cal B$ on $A\cup B$. By Proposition~\ref{prop:blockgraphtree} we know that $\cal B$ is a tree. Now we root this tree at some vertex in $B$. Observe that there is unique association of cut vertices to its parent -- which is a block.  
We also know that for every  cut vertex $v$ that either $v$ is in $S$ or a vertex in its parent  block. However, the blocks are pairwise disjoint except for the vertices in $A$. Thus, this implies that there is an injective map from $A$ to $S$ and hence $|{\sf CutVert}|\leq |S|$. 
\end{proof}

Now we are ready to prove the treewidth bounding lemma of this section. 
 Just as for \tDS, it is possible to prove that after removing all irrelevant vertices, the treewidth of each torso in the reduced graph is $\cO(\sqrt{|S|})$. The most important difference is that instead of $2$-dominating set we construct a $8h$-dominating set in the proof.   We start with the following auxiliary lemma that will be useful for the proof.

\begin{lemma}
\label{lem:sumreductioncondomset}
There is a polynomial time algorithm that for  a given graph $G$ and a dominating set $S$ of $G$, outputs graph $G'$ such that $G' \equiv_{\tCDS} G $ and for every torso $\tau(t)$ of the tree-decomposition $(M,\Psi)$ of $G$, we have that $\tau(t)\setminus Z_t$ has a $8h$-dominating set of size $\cO(|S|)$. Furthermore if $G$ is a \Hmf \, graph then $\tw(G)=\cO(\sqrt{|S|})$. 
\end{lemma}
\begin{proof}
Let $\tau(t)^*= \tau(t) \setminus A$, where $A$ are the apices of $\tau(t)$. Also, let ${\sf CutVert}$ denote the set of cut vertices of $G$. 
We will obtain a $(4h+1)$-dominating set of size $\cO(|S|)$ in $\tau(t)^*$. Towards this end, consider the following set, 
$$Q=\bigcup_{A' \subseteq A, A' \text{is feasible} }D(A')\cup R\cup (S \setminus A) \cup {\sf CutVert} .$$
The size of the set of representative vertices, $R$, is at most $4h \cdot 2^{|A|}\leq 4h \cdot 2^h$. The number of feasible subsets $A'$ is at most $2^h$, where $h$ is a constant depending only on $H$. The size of $D(A')$ is at most $\eta(H)(3|S|+3)$ for every $A'$. By Lemma~\ref{lem:cutvertexcount} we have that $|{\sf CutVert}|\leq |S|$. Thus  $|Q|\leq 2^h (\eta(H)(3|S|+3)) + 4h \cdot 2^h + 2 |S|=\cO(|S|)$. We prove that $Q$ is a $(4h+1)$-dominating set of $V(G) \setminus A$. Let $w \in V(G) \setminus A$. If $w \in R$ or $w \in S$ or $w\in {\sf CutVert}$ then $Q$ dominates $S$. So suppose $w \notin R \cup S \cup {\sf CutVert}$. Then, since $w$ is not irrelevant there is a feasible subset $A'$ of $A$ such that $w$ is relevant with respect to $A'$. Hence there exists a vertex $w'$ in $N_{G\setminus A}^{4h}[w]$ which is not in $W(A')$. If $w' \in S_{4h}$, $S_{4h}$ denotes the set of vertices in $N_{G\setminus A}^{4h}[S]=N_{G\setminus A}^{4h}[S\setminus A]$, 
 then $w$ is $8h$-dominated by a vertex  $w^* \in (S\setminus A) \subseteq Q$ in $G \setminus A$. Otherwise $w'$ is dominated by some $w''$ in $D(A')$ and hence $w$ is $4h+1$-dominated by $w'' \in Q$ in $G \setminus A$. Hence $G \setminus A$ has a $8h$-dominating set of size $\cO(|S|)$.

The graph $\tau(t)^*$ can be obtained from $G \setminus A$ by contracting all edges in $E(G \setminus A) \setminus E(\tau(t)^*)$ and adding all edges in $E(\tau(t)^*) \setminus E(G \setminus A)$. Since contracting and adding edges can not increase the size of a minimum $8h$-dominating set of a graph, $\tau(t)^*$ has a $8h$-dominating set of size $\cO(|S|)$.
This completes the proof for the first part.

Now assume that $G$ is a \Hmf \, graph. It is well known  that the treewidth of a \Hmf \, graph is at most the maximum treewidth of its torsos, see e.g.\cite{DemaineFHT05sub}. Thus to show that $\tw(G)=\cO(\sqrt{|S|})$ it is sufficient to show that its torsos have small treewidth.  
To conclude, $\tau(t)^*$ excludes an apex graph as a minor  (see discussions after Theorem~\ref{thm:structure theorem}) and it has a $8h$-dominating set of size $\cO(|S|)$. By the bidimensionality of $8h$-dominating set, we have that $\tw(\tau(t)^*)=\cO(\sqrt{|S|})$~\cite{DemaineFHT05sub,FominGT09con}. Now we add all the apices of $A$ to all the bags of the tree-decomposition  of $\tau(t)^*$ to obtain a tree-decomposition for $\tau(t)'$.  Thus $\tw(\tau(t)')\leq \cO(\sqrt{|S|})+h =\cO(\sqrt{|S|})$. 

Let us also remark that Irrelevant Vertex Rule is based on the performance of a polynomial time approximation algorithm and thus the whole procedure can be implemented in polynomial time. This concludes the proof. 
\end{proof}



Having Lemma \ref{lem:sumreductioncondomset} proving Lemma~\ref{lem:newperspectivequasicds} becomes simple.

\begin{proof}[Proof of Lemma~\ref{lem:newperspectivequasicds}]
We apply Lemma~\ref{lem:sumreductioncondomset} to $G$ with a decomposition that has a single bag containing the entire graph and the apices $A$ of the bag being the vertices of degree at least $h'$. By Lemma~\ref{lem:sumreductioncondomset}, $G \setminus A$ has a $8h$-dominating set of size $\delta_3|S|$. Since all vertices of $G \setminus A$ have degree at most $h'$ it follows that $|V(G)| \leq h' + h'^{\cO(h')}\delta_3 |S| 
\leq \delta|S|$.  
\end{proof}


Proof for Lemma~\ref{lem:newperspectiveHminorcds} is identical to the proof of Lemma~\ref{lem:newperspectiveHminor}, except that we need to use Lemma~\ref{lem:sumreductioncondomset} in  
place of Lemma~\ref{lem:sumreductiondomset}. Thus we omit it. 
 

Recently, Bodlaender et al.~\cite{BodlaenderCKN13} obtained an algorithm solving \tCDS \ on graphs of treewidth $t$ in time $c^{t}n^{\cO(1)}$. Theorem~\ref{thm:lineardomsettopocds} combined with this implies that  \tCDS \,  on \Hmf \, graphs is solvable in time $2^{\cO(\sqrt{k})}+ n^{\cO(1)}$. To our knowledge, this is the first subexponential parameterized algorithm for \tCDS \, on \Hmf \, graphs.


\begin{theorem}
\label{thm:subexpcondomset}
Given an $n$-vertex  graph $G$ excluding a fixed graph $H$ as a minor,  one  can check whether $G$ has a connected dominating set of size at most $k$ in time 
$2^{\cO(\sqrt{k})}+ n^{\cO(1)}$. 
\end{theorem}

\section{Conclusions}\label{sec:concludes}
In this paper we give linear kernels for two widely studied parameterized problems, namely \tDS\ and  \tCDS,
for every graph class that excludes some graph as a topological minor. The emerging questions are the following two:

\begin{enumerate}
\item Can our kernelization results for  \tDS\ and  \tCDS\  be extended to more general sparse graph classes?
\item Can our techniques be applied to more general families of parameterized problems?
\end{enumerate}

%

Very recently,  the first question was answered both positively and negatively by Drange et al. \cite{Drange2015}.  In particular, \tDS \ admits a vertex-linear kernel on graphs of bounded expansion and an almost vertex-linear kernel on nowhere-dense graphs. On the other hand  
  \tCDS  \ admits no polynomial kernel on graphs of bounded expansion unless \textsf{coNP} 
  $\subseteq$  \textsf{NP/poly}.  It is important to point out that methods used by Drange et al. \cite{Drange2015} is entirely different than ours. Their algorithm is completely combinatorial and do not rely on topological arguments. Our kernelization algorithm for \tCDS\ is still the best known. It would be interesting to see if the combinatorial methods developed in 
  Drange et al. \cite{Drange2015} could be used to design an explicit kernelization algorithm for \tCDS\ on graph classes excluding a fixed graph $H$ as a topological minor.

\paragraph{Acknoweldgements}
Thanks to Marek Cygan for sending us a copy of \cite{CyganGH12}. We sincerely thank all the reviewers for their insightful comments and suggestions. 

%
%

\end{document}